\documentclass[11pt]{article}
\usepackage[letterpaper,margin=1in]{geometry}

\usepackage[utf8]{inputenc}
\usepackage{amsmath}
\usepackage{amsfonts}
\usepackage{amssymb}
\usepackage{amsthm}
\usepackage{graphicx}
\usepackage{fancyhdr}
\usepackage{algorithm}
\usepackage{algorithmicx}
\usepackage[noend]{algpseudocode}
\usepackage{booktabs}
\usepackage{makecell}
\usepackage[toc,page]{appendix}
\usepackage{lipsum}
\usepackage{cite}
\usepackage{listings}
\usepackage{xcolor}
\usepackage{todonotes}
\usepackage{siunitx}
\usepackage{aligned-overset}
\usepackage{hyperref}
\usepackage{cleveref}
\DeclareMathOperator*{\argmin}{arg\; min}
\DeclareMathOperator*{\argmax}{arg\; max}
\DeclareMathOperator*{\minimize}{minimize\ }
\DeclareMathOperator*{\maximize}{maximize\ }
\DeclareMathOperator*{\tr}{trace }
\DeclareMathOperator*{\diag}{diag }

\newcommand{\tF}{\Tilde{F}}
\newcommand{\tOmega}{\Tilde{\Omega}}
\newcommand{\rev}[1]{\textcolor{black}{#1}}
\newcommand{\rf}[1]{\textcolor{black}{#1}}

\newtheorem{lemma}{Lemma}

\let\arxivbibliographystyle\bibliographystyle
\renewcommand{\bibliographystyle}[1]{\arxivbibliographystyle{plain}}

\begin{document}
% Shared content for main.tex and template.tex. Edit the manuscript here.
\title{Enhancing a Risk Model by Adding Transient Statistical Factors}

\author{Alexandros E. Tzikas\textsuperscript{1,2}\and Emmanuel J. Candès\textsuperscript{1,2}\and Trevor Hastie\textsuperscript{1,2}\and Stephen P. Boyd\textsuperscript{1,2}\and Mykel J. Kochenderfer\textsuperscript{1,2}\and Ronald N. Kahn\textsuperscript{2}}

\ifdefined\institute
\institute{\textsuperscript{1}\at
             Stanford University\\
             Stanford, CA 94305, USA \\
            %  \email{\{alextzik,candes,hastie,boyd,mykel\}@stanford.edu}
           \and
              \textsuperscript{2} \at
              BlackRock \\
              400 Howard Street \\
              San Francisco, CA 94105, USA\\
            %   \email{ron.kahn@blackrock.com}
}

% \date{Received: 13 May 2026}
%The correct dates will be entered by the editor.

\else
\date{\small
\textsuperscript{1}Stanford University, Stanford, CA 94305, USA\\
% \texttt{\{alextzik,candes,hastie,boyd,mykel\}@stanford.edu}\\
\textsuperscript{2}BlackRock, 400 Howard Street, San Francisco, CA 94105, USA\\
% \texttt{ron.kahn@blackrock.com}\\[4pt]
% Received: 13 May 2026}
}
\fi

\maketitle

\begin{abstract}
    Estimating the covariance of asset returns, i.e., the risk model, is a key component of financial portfolio construction and evaluation.
    % An accurate risk model allows investors to build risk-aware portfolios and assess the risk level of already constructed ones.
    Most risk modeling approaches produce a \rf{factor} model that decomposes the asset variability into two components: the first attributed to a small number of factors that are common among the assets and the second attributed to the idiosyncratic behavior of each asset.
    % \rf{Factor} models are both interpretable and computationally efficient when used for portfolio construction.
    Third-party providers typically provide risk models to investors, and while these models are typically of high quality, they may fail to capture important information, e.g., changing market regimes and transient factors.
    % \rev{or may not be tuned to the desired asset universe}.
    % \ejc{I do not understand the last sentence. What do we mean it is not tuned to the asset universe?}
    To overcome these limitations, we propose a systematic method based on maximum likelihood estimation to enhance \rf{an existing factor} model by \rf{both} refining the given model and adding new statistical factors. Our approach relies only on the observed sequence of realized returns and on the choice of two hyperparameters: the number of additional factors and the half-life parameter that determines the weights assigned to returns in the log-likelihood objective.
    \rf{Importantly, our methodology applies to the situation where asset returns may be missing, making it suitable for typical equity datasets.}
    \rev{We demonstrate our approach on the Barra short-term US risk model, a high-quality risk model used in practice, for a universe of US high-capitalization equities. We show that the proposed extension captures structure in the returns that is missed by the original model.}
    % \ejc{The abstract is fine but 50\% is dedicated to introducing factor models which are well understood. Comparably little time is spent discussing what the paper actually does, which should be the point of an abstract. (fixed, I think).}
\end{abstract}

% \keywords{factor model \and maximum likelihood estimation \and statistical factors}
% \subclass{49J53 \and  49K99 \and more}

\section{Introduction}\label{sec:intro}
    The covariance of asset returns is a central quantity in several areas of finance, including Markowitz portfolio construction \cite{markowitz1952, grinold2000active, boyd2024markowitz}, risk management \cite{mcneil2015quantitative}, asset pricing \cite{sharpe1964capital}, \rf{and performance attribution analysis \cite{menchero2008custom}}.
    % In this paper, we call an asset any financial instrument that is part of a portfolio or investment strategy, such as a stock, bond, or exchange traded fund (ETF), among others.
    % \rev{In this paper, we use the term asset to broadly refer to an individual tradable financial instrument that forms one element of a portfolio or investment strategy, such as a stock, bond, or exchange-traded fund (ETF).}
    % The risk of a portfolio (or investment strategy) is defined as the variance of its return, where the portfolio return is a linear function of the asset returns vector \cite{boyd2024markowitz}.
    \rf{The} \rf{variance} of a portfolio (or investment strategy) is a linear function of the asset return covariance. An accurate \rev{(as judged by its statistical fit)} asset return covariance allows investors to build portfolios with a desired level of risk and also evaluate the risk of existing portfolios. An inaccurate risk model can over-estimate or under-estimate portfolio risk. As a result, during portfolio construction, it may be overly restrictive, causing the investor to miss out on potentially good investments, or fail to account for certain risk sources, potentially exposing the investment strategy to unexpected risks.
Risk modeling is therefore a challenging problem in any investment process.

A challenge is that only information available at the current time can be used to estimate the risk model, which should be predictive of the assets' behavior in the future \rf{(that is, out-of-sample)}. \rev{Prediction} is particularly difficult because of changing market conditions \cite{stuaricua2005nonstationarities}, extreme events (e.g., COVID-19), and fat tails in the distribution of returns \cite{fan2016overview}. In addition, asset returns \rev{can fluctuate significantly} even on a daily basis, making it difficult to separate meaningful structure from randomness.
The high dimensionality of the data, \rev{i.e., the large number of assets (possibly thousands)}, is another challenge \cite{fan2008high}.  \rf{Capturing the dependence structure} of thousands of assets requires simultaneously and accurately estimating a very large number of parameters from return data. For instance, the sample covariance based on the observed data is singular
when the number of assets is larger than the number of sample return vectors \cite{fan2016overview}. \rev{This implies that there are directions in asset space with zero estimated variance, and it is easy to understand that this will cause problems in portfolio construction.} Because of these challenges, various companies, for example MSCI Barra \cite{msci_barra2023, menchero2011barra}, offer risk models as products and a Nobel Memorial Prize in Economic Sciences was awarded for work directly related to volatility estimation \cite{engle1982autoregressive}.

A standard model for the asset returns is the \rf{factor} model. \rf{At any time $t$}, it assumes that the vector of asset returns, $x_t \in \mathbb{R}^n$, is driven by a low-dimensional vector of factor returns, $s_t \in \mathbb{R}^{n_1}$, and a vector of residual returns $\epsilon_t \in \mathbb{R}^n$, i.e.,
\begin{equation*}
    x_t = F_t s_t + \epsilon_t,
\end{equation*}
where $F_t \in \mathbb{R}^{n \times n_1}$ is the matrix of factor exposures, $\epsilon_t$ is independent of $s_t$, the \rev{$n$ elements of $\epsilon_t$ are independent}, and $n_1 \ll n$ \cite{grinold2000active}. \rf{In this model, $x_t$, $s_t$, and $\epsilon_t$ are random vectors and $F_t$ is a non-random matrix that can vary with time $t$.} The model implies that there are common factors that influence the returns of the assets. For equities, factors typically include the returns of specific portfolios, such as the
overall market (with weights proportional to capitalization), industries, and style portfolios like the Fama-French factors \cite{fama1992cross, fama1993common}. The model also implies that there is a part of the returns that cannot be attributed to these factors. We call this part the idiosyncratic return, $\epsilon_t$, which is independent of the factor returns, $s_t$, and also independent for each asset.

Under these assumptions the asset return covariance \rf{is given by the low-rank-plus-diagonal form}
\begin{equation}\label{eq:risk_model}
    \mathbf{cov}(x_t) = F_t \mathbf{cov}(s_t) F_t^\top + \mathbf{cov}(\epsilon_t),
\end{equation}
where $\mathbf{cov}(\cdot)$ denotes the covariance. Because the \rev{elements} of $\epsilon_t$ are independent, $\mathbf{cov}(\epsilon_t)$ is a diagonal matrix with positive diagonal entries.
This structure is especially appealing in high dimensions ($n$ in the order of hundreds or thousands), because we only need to estimate $\mathcal{O}(n n_1)$ parameters, instead of $\mathcal{O}(n^2)$ in the case of a generic $n \times n$
covariance matrix. The smaller number of parameters required in the low-rank-plus-diagonal structure \rf{\eqref{eq:risk_model}} is also a form of regularization, which avoids overfitting and can improve the covariance estimate, especially for out-of-sample applications (such as predicting future portfolio risk) \cite{fan2008high}. The low-rank-plus-diagonal structure also allows for a positive definite covariance matrix even when \rev{there are fewer than} $n$ return samples. Finally, in the context of portfolio construction using convex optimization, the low-rank-plus-diagonal structure can be exploited to bring the computational complexity down from $\mathcal{O}(n^3)$ to $\mathcal{O}(n n_1^2)$ operations \cite{boyd2024markowitz}.

In this paper, we assume a low-rank-plus-diagonal risk model is provided. We call this the \textit{base} or \textit{existing} model. This scenario is typical in the case of investors or quantitative traders, e.g., an MSCI Barra \cite{msci_barra2023, menchero2011barra} model is available. We propose a principled method based on maximum likelihood estimation to refine the existing model and extend it with additional \rev{statistical} factors using the time series of realized asset returns. Our method can be used to capture market shifts of a shorter time horizon than the update period of the provided risk model.
% \rev{or tune the model to a desired asset universe}.
Furthermore, the added statistical factors capture structure in the returns that is missed by the base model.
\rev{As a demonstration of our approach, we extend the Barra short-term US risk model (already high-quality and widely used in practice) for a universe of 870 US high-capitalization equities and obtain an improvement in terms of the model's out-of-sample statistical fit. We show that the extended model uncovers hidden structure in the returns that is missed by the base model.}
\rev{As a secondary contribution, we also show how to handle missing return data within our framework. Missing returns arise in financial datasets for several practical reasons, such as reporting delays, non-synchronous trading across time zones, or unavailable price observations on certain dates.}

The remainder of the paper is organized as follows. In \Cref{sec: prior_work}, we discuss related work in covariance estimation and low-rank-plus-diagonal risk modeling. In \Cref{sec:problem_form}, we mathematically formulate our maximum likelihood estimation problem. Our proposed approach is described in \Cref{sec:proposed}. We discuss metrics and results in \Cref{sec:resultss}. The paper concludes with \Cref{sec:conclusion}.

\section{Prior Work}\label{sec: prior_work}
We review \rf{in \Cref{sec:empirical_cov}} the standard approaches to empirical covariance estimation that rely solely on the time series of historical returns and do not impose any structural assumptions.
Next, \rf{in \Cref{sec:lrpd}} we turn to techniques that introduce additional structure into the covariance matrix—most notably low-rank-plus-diagonal decompositions—which aim to capture common risk factors and improve stability in high-dimensional settings.

\subsection{Empirical Covariance Estimation}\label{sec:empirical_cov}
Given a time series of historical returns, the simplest covariance estimate is the sample covariance. However, when the number of assets
is large relative to \rev{the number of return samples} (which is almost always true in practice), the sample covariance is a very poor approximation of the true covariance \cite{pafka2003noisy} and its \rf{eigenvalues can have high variance}. Generally, the sample covariance may fail to capture the time-dependence and regime shifting of financial markets.
A simple way to deal with the regime shifting is to use a rolling window estimate, which has a fixed memory window and only averages the return outer products in this window. To improve the quality of the estimator we can add regularization or shrinkage \cite{ledoit2003honey}.
Practitioners often replace rolling window covariance estimates with exponentially weighted moving averages (EWMAs) to place greater weight on recent observations, thereby adapting more quickly to time-varying conditions \cite{johansson2023simple, menchero2011barra, jp1996riskmetrics}.
% However, \rev{a fixed window can introduce abrupt changes in} covariance estimates and \rev{can make} the estimator highly sensitive.
% \ejc{Which estimator is sensitive? Also, I don't think that if I have a rolling window of 90 days, a shift of 1 day will cause an abrupt change. I am not even sure I understand why it will cause a larger change than the EWMA. (fixed)}
% The exponentially weighted moving average (EWMA) estimator is the smooth equivalent of the rolling window estimator and better handles the non-stationarity of financial returns \cite{barratt2023covariance, engle2002dynamic}.
The EWMA estimator is controlled by the half-life parameter that dictates how far in the past the weight of an observation is reduced to $0.5$. The half-life parameter is thus completely interpretable. The EWMA estimate also allows for fast, recursive computation.

When the number of assets is large and the half-life (or the length of the return time series) is small, naive covariance estimation leads to singular matrices. In these cases, even though it might fail to capture information in the returns, a popular remedy is to learn a factor model \cite{grinold2000active}, which we discuss next.

\subsection{Low-Rank-plus-Diagonal Risk Models}\label{sec:lrpd}
\rev{Factor models are a natural way to represent correlations among assets, because assets in similar market sectors tend to behave similarly\rf{~\cite{rosenberg1974extra}}.}
We discuss methods that produce risk models in low-rank-plus-diagonal form. These methods differ in how they obtain the factor exposure matrix, the covariance of the factor returns, and the covariance of the idiosyncratic returns. \rev{Each method corresponds to a different category of factors: fundamental, macroeconomic, and statistical.}
% MK: drop the "f we consider the nature of the involved factors" (might be a little confusing) (Fixed)
% There are \rev{three} groups of methods if we consider the nature of the involved factors: fundamental, \rev{macroeconomic}, and statistical \cite{grinold2000active}.
\rev{Fundamental and macroeconomic} factors are built from observable economic, financial, or asset-specific characteristics. Such a factor can correspond to a concept that has an intuitive interpretation grounded in economics, corporate structure, or investor behavior \cite{fabozzi2008handbook}. \rev{Fundamental factors include} industry/sector classifications, \rev{while macroeconomic factors include} \rf{inflation, unemployment, industrial productivity, as well as oil, gold, and interest rates}.
% The \rev{macroeconomic} Fama-French factors \cite{fama1992cross, fama1993common} are built from returns on value-weighted or equal-weighted portfolios formed from characteristics such as size, book-to-market, profitability, or investment.
Statistical factors are extracted directly from historical return data through dimensionality-reduction techniques \cite{rubin1982algorithms, seghouane2008iterative}. They attempt to capture the main directions of covariance in the return vector without imposing any economic interpretation.

\paragraph{Methods Based on Fundamental Factors} \rev{In this case, we calculate the factor exposure matrix and then estimate the factor returns by regressing the asset returns on the factor exposures \cite{nielsen2010fundamentals}. \rf{Either ordinary or generalized least squares (where we scale each asset by its volatility) is used in this empirical approach.} We then estimate the covariance of the factor returns using a method from \Cref{sec:empirical_cov} and the variances of the idiosyncratic returns similarly from the regression residuals. We mention an example of a fundamental factor: the price-to-earnings ratio. We can assign the price-to-earnings ratio of each asset as its exposure to the factor, perhaps with some cross-sectional (i.e., across assets) standardization.}

\paragraph{Methods Based on Macroeconomic Factors} \rev{In this case, we calculate the factor returns first. For example, we can use the returns of style portfolios as the factor returns. We then compute the factor exposure matrix by regressing the historical asset returns on the historical factor returns \cite{hastie2009elements}. We estimate the covariance of the factor returns using a method from \Cref{sec:empirical_cov} and the variances of the idiosyncratic returns similarly from the regression residuals. Note that the estimation of the factor exposure matrix is separable across assets. This is also an empirical approach.}

\paragraph{Methods Based on Statistical Factors} Such methods start by computing an empirical asset return covariance, as described in \Cref{sec:empirical_cov}. The low-rank-plus-diagonal structure is imposed by dimensionality reduction. Effectively, we find the low-rank-plus-diagonal model that optimizes an objective involving the empirical covariance.
We can learn the low-rank component by minimizing its Frobenius norm to the empirical covariance. The solution is related to the top eigenvalues (and corresponding eigenvectors) of the empirical estimate \cite{bai2008large}, \cite[Section 8.2]{johansson2023simple}, \cite{yeon2025beyond}. We can then set the idiosyncratic component of the risk model to match the diagonal entries of the empirical covariance estimate.
This approach tends to underestimate the idiosyncratic risk.
If we fit the risk model by maximizing Gaussian log-likelihood, then the solution can be obtained via an iterative algorithm called expectation-maximization (EM) \cite{johansson2023simple, rubin1982algorithms, seghouane2008iterative} or an alternating minimization approach \cite{mardia2024multivariate}.
% We note that the EM algorithm is able to handle missing data in a principled fashion.
% \rf{A Frobenius-norm fit primarily tries to reproduce the largest entries and dominant eigendirections of the empirical covariance matrix, whereas maximum-likelihood fitting balances errors in all directions, including low-variance ones.}

\rev{The risk model is time-dependent as the factor exposures, factor return covariance, and idiosyncratic return covariance are all time-dependent quantities. The update frequency of the risk model can vary from daily to, e.g., quarterly, depending on the investor's needs and the nature of the factors involved. The update frequency can be a limiting factor when it comes to the responsiveness of the risk model to changing market conditions. For example, a quarterly update frequency may not be able to capture sudden market regime shifts or transient factors that appear on shorter time scales.}

\paragraph{Our Contribution} We propose a method that can extend any given low-rank-plus-diagonal risk model, either fundamental, \rev{macroeconomic}, or statistical. Our algorithm uses the history of observed returns and a weighted Gaussian log-likelihood objective to refine the base model and extend it with additional statistical factors. As such, it is a principled way to build hybrid models with both fundamental-macroeconomic and statistical factors. Spector et al. \cite{spector2024mosaic} investigated how to assess the statistical fit of an existing risk model and proposed a simple method to extend the model by recursively adding factors. Our work is similarly based on a provided model, but our method both refines the given risk model and adds new statistical factors.
\rf{Lee et al. \cite{lee2025narrative} also propose extending a risk model by additional factors. However, the asset exposures to these new factors are not estimated statistically, but are instead constructed.
The new factor exposures are thematic narrative exposures built from news data. They are constructed based on how frequently the narrative appears in news coverage about each asset. Furthermore, the covariance of the new factor returns is estimated following an empirical approach that uses ordinary least squares.
In contrast, our approach relies on statistically learning the additional factors in order to improve the likelihood fit of the observed data. In addition, beyond adding new factors, our approach also refines the covariance of the factor returns of the base model}.

We leave the algorithmic selection of the number of additional factors as future work. \rev{We note that our approach preserves the factor exposures to the base factors. As such, the methodology is particularly
suitable when the base model consists of fundamental factors. \rf{By choosing the log-likelihood weights based on an EWMA with an appropriate (typically short) half-life, our work can also be regarded as a principled way to find additional} \rf{transient} factors in the returns, which are commonly known as themes \cite{candes2025thematic}.}

\section{Problem Statement}\label{sec:problem_form}

% \ejc{I made changes throughout this section.}
\rf{Suppose at time $t$ we are given a sequence of asset return vectors $x_1, \dots, x_T \in \mathbb{R}^n$, potentially with missing entries. We are also provided with the \textit{base} asset return covariance model}
\begin{equation*}
  \mathbf{cov}(x_t) \, \approx \, F_1 \Omega_{\mathrm{base}} F_1^\top + D_\mathrm{base},
\end{equation*}
where $F_1 \in \mathbb{R}^{n \times n_1}$ is the matrix of factor exposures to the \textit{base} factors, $\Omega_\mathrm{base}$ is the (positive definite) covariance of the \textit{base} factor returns, and $D_\mathrm{base}$ is a diagonal matrix with positive diagonal entries corresponding to the idiosyncratic asset variances.
In practice, the provided covariance model depends on time $t$, i.e., $F_1$, $\Omega_\mathrm{base}$, and $D_\mathrm{base}$ depend on $t$. For simplicity we hide this dependence.

% \rf{We begin by describing the structure of the risk model to be estimated from the observed sequence of returns. This learned model also depends on time $t$ but we hide this dependence. We will call this the \textit{extended} risk model, because it involves additional learned factors. We then present the objective function used for estimation.}

\subsection{The Extended Risk Model}
\rf{We consider the \textit{extended} asset return covariance model} (this model also depends on time $t$ but we hide this dependence)
\begin{equation}\label{eq:family1}
    % \hat{\Sigma}\left(F_2, \Sigma_f, D\right)
    \mathbf{cov}(x_t) \, \approx \,  \begin{bmatrix}
        F_1 & F_2
    \end{bmatrix}\Sigma_f\begin{bmatrix}
        F_1^\top\\
        F_2^\top
    \end{bmatrix} + D,
\end{equation}
where $F_2 \in \mathbb{R}^{n \times n_2}$, $\Sigma_f \in \mathbb{R}^{(n_1+n_2) \times (n_1+n_2)}$, and $D \in \mathbb{R}^{n \times n}$. $D$ is a diagonal matrix with positive diagonal entries and $\Sigma_f$ is a positive definite matrix.
The number of \textit{base} factors is $n_1$ and the number of added factors is $n_2$.
Typically the total number of factors, $n_1 + n_2$, is much smaller than $n$.
% \rf{Although the risk model $\hat{\Sigma}\left(F_2, \Sigma_f, D\right)$ depends on $t$, we hide the dependence for simplicity in this and the next section.}
We look to find appropriate factor exposures $F_2$, factor return covariance $\Sigma_f$, and idiosyncratic covariance $D$.

We make a distinction between the factors involved in $F_1$ and $F_2$.
\begin{itemize}
    \item $F_1$ is a known matrix of factor exposures, which have been measured and provided to us. For example, $F_1$ could be the matrix of factor exposures from an MSCI Barra model \cite{msci_barra2023}. However, the covariance of the factor returns involved in $F_1$ will be re-estimated, by learning $\Sigma_f$. The only information from the \textit{base} model that the \textit{extended} model shall rely on without modification is the matrix of factor exposures $F_1$.
    This is natural if the base model consists of fundamental factors: in such a model exposures are constructed from relatively stable asset characteristics or classification rules, while the covariance of factor returns is computed empirically.
    % This implies that in the base model we trust the factor exposures more than the factor returns (that were used to estimate the factor return covariance). This can happen if, for example, the factor exposures were first determined and then regression was used to compute the factor returns for the base model, \rev{i.e., fundamental factors}.
    % \ejc{I am not sure I understand the argument which states that $F_1$ is more reliable than $\Omega$ especially in light of what is stated in Section 2. This is a chicken and egg problem. (fixed I think, if not we should just remove)}
    % \rev{This makes it easy to include desired  factors in our risk model. }
    \item \rf{$F_2$ denotes a matrix of unknown factor exposures representing residual factor structure not captured by $F_1$. Estimating $F_2$ and $\Sigma_f$ therefore allows the model to account for additional common sources of variation in asset returns beyond those included in the base model.}
\end{itemize}

\subsubsection{Identifiability of the Extended Risk Model}
Consider first a general factor model and the implied covariance $F \Sigma_f F^\top +D$, where all the terms are unknown. This is different from our setting \eqref{eq:family1}, where $F_1$ is known. In the general model, neither $F$ nor $\Sigma_f$ are identifiable, since for any invertible matrix $R$, setting
\[
\bar{F} = FR,\quad \bar{\Sigma}_f = R^{-1} \Sigma_f R^{-\top}
\]
yields the same covariance.
With the Cholesky decomposition $\Sigma_f = L_f L_f^\top$,  another equivalent form is
\[
\bar{F} = F L_f, \quad \bar{\Sigma}_f = I.
\]
Since the latter is the easiest representation, the standard form for the general model assumes $\Sigma_f = I$.

Our case \eqref{eq:family1} is a little more complex, because $F_1$ is known. We show that the analogous representation for the family of risk models \eqref{eq:family1} is
\begin{equation}\label{eq:family2}
     \begin{bmatrix}
        F_1 & F_2
    \end{bmatrix}\begin{bmatrix}
        \Omega & 0\\
        0 & I
    \end{bmatrix}\begin{bmatrix}
        F_1^\top\\
        F_2^\top
    \end{bmatrix} + D,
\end{equation}
where $\Omega \in \mathbb{R}^{n_1 \times n_1}$ is a positive definite matrix,
$F_2 \in \mathbb{R}^{n \times n_2}$ is arbitrary and $D$ is diagonal.
Note that $F_2$ is still only identifiable up to a rotation/reflection: if $R$ is any $n_2 \times n_2$ orthogonal matrix, then the representation \eqref{eq:family2} with $F_2$ replaced by $F_2 R$ is equivalent, since $F_2 R R^\top F_2^\top = F_2 F_2^\top$ (a standard situation in factor models).
We prefer the representation given by the family of models \eqref{eq:family2}, because it ensures that the additional factors remain identifiable (up to a rotation/reflection) and allows for a decomposition of risk across the \textit{base} factors $F_1$ and the added factors $F_2$.
Family \eqref{eq:family2} involves a smaller number of parameters compared to family \eqref{eq:family1}, i.e., family \eqref{eq:family1} is over-parameterized.

\begin{lemma}
The family of models  given by the right-hand side of \eqref{eq:family1} is the same as the family of models given by \eqref{eq:family2}.
\end{lemma}
\begin{proof}
    Consider a model $\hat{\Sigma}\left(F_2, \Omega, D\right)$ in the family \eqref{eq:family2}. Obviously, this model also belongs to the family of models given by \eqref{eq:family1}. Simply set
    \begin{equation*}
        \Sigma_f = \begin{bmatrix}
        \Omega & 0\\
        0 & I
    \end{bmatrix}.
    \end{equation*}
Now, consider the model $\hat{\Sigma}\left(F_2, \Sigma_f, D\right)$ in the family \eqref{eq:family1} and the decomposition
\begin{equation*}
    \Sigma_f = \begin{bmatrix}
        \Sigma_{f, 11} & \Sigma_{f, 12} \\
        \Sigma_{f, 21} & \Sigma_{f, 22}
    \end{bmatrix},
\end{equation*}
where $\Sigma_{f, 11} \in \mathbb{R}^{n_1 \times n_1}$. Let
\begin{equation*}
    \Phi = \begin{bmatrix}
        I & \Sigma_{f, 12} \Sigma_{f, 22}^{-1/2}\\
        0 & \Sigma_{f, 22}^{1/2}
    \end{bmatrix}, \quad \Omega = \Sigma_{f, 11} - \Sigma_{f, 12} \Sigma_{f, 22}^{-1} \Sigma_{f, 21}.
\end{equation*}
The matrix $\Omega$ is positive definite as the Schur complement of $\Sigma_f$. Because $\Sigma_f$ is positive definite, $\Sigma_{f, 22}$ is also positive definite.

We can write
\begin{equation}
\label{eq:ortho}
\begin{aligned}
    &\hat{\Sigma}\left(F_2, \Sigma_f, D\right) = \begin{bmatrix}
        F_1 & F_2
    \end{bmatrix} \Phi \begin{bmatrix}
        \Omega & 0 \\
        0 & I
    \end{bmatrix}\Phi^\top \begin{bmatrix}
        F_1^\top \\
        F_2^\top
    \end{bmatrix} + D\\
    &= \begin{bmatrix}
        F_1 & F_1 \Sigma_{f, 12} \Sigma_{f, 22}^{-1/2} + F_2 \Sigma_{f, 22}^{1/2}
    \end{bmatrix}
    \begin{bmatrix}
        \Omega & 0 \\
        0 & I
    \end{bmatrix}
    \begin{bmatrix}
        F_1 & F_1 \Sigma_{f, 12} \Sigma_{f, 22}^{-1/2} + F_2 \Sigma_{f, 22}^{1/2}
    \end{bmatrix}^\top + D.
\end{aligned}
\end{equation}
Therefore, the model $\hat{\Sigma}\left(F_2, \Sigma_f, D\right)$ in the family \eqref{eq:family1} can be represented by a model in the family \eqref{eq:family2}.
\end{proof}

From \eqref{eq:ortho}, it becomes evident that the columns of the added factor exposure matrix $F_2$ are not necessarily orthogonal to the columns of $F_1$. There exists another representation that is equivalent to \eqref{eq:family1} and \eqref{eq:family2}, but constrains the new factor exposure matrix $F_2$ to have orthogonal columns to $F_1$. Although not necessary for the purposes of this paper, the proof is given in \Cref{appendix:param}. Lee et al. \cite{lee2025narrative} extend a provided risk model with $F_2$ that has orthogonal columns to $F_1$. In the remainder of the paper, we will consider the representation given by family \eqref{eq:family2} for the \textit{extended} model.

\subsection{Estimation Objective}

% \ejc{I don't like the notation $\tau \in \{1, \ldots, t\}$. I prefer $t \in \{1, \ldots, T\}$. $t$ is used throughout to refer to the current time and as you note it does not need to match the number of samples. I would be hesitant to use it as a rolling index as well when we have $\tau$ available. I made a compromise between the two: $\tau \in 1, \dots, T$. Let me know if it works}

We assume that at each time $\tau=1, \dots, T$ we only observe the returns $x_\tau$ on a subset of assets indexed by the set $\mathcal{O}_\tau \subseteq \{1, \dots, n\}$. Let $\mathcal{M}_\tau = \{1, \dots, n\} \setminus \mathcal{O}_\tau$ be the set of assets with missing returns at time $\tau$. We denote the observed portion of the return at $\tau$ as $x_{\tau, \mathrm{obs}}$ and the missing portion as $x_{\tau, \mathrm{mis}}$.

Let $p_\theta(x)$ be the probability density function of the Gaussian $\mathcal{N}(0, \tF \tOmega \tF^\top +D)$, where
\begin{equation*}
    \tF = \begin{bmatrix}
        F_1 & F_2
    \end{bmatrix},\qquad \tOmega = \begin{bmatrix}
        \Omega & 0 \\
        0 & I
    \end{bmatrix},
\end{equation*}
and $\mathcal{N}(\mu, \Sigma)$ is the Gaussian distribution with mean $\mu$ and covariance $\Sigma$.
The covariance of the Gaussian $\mathcal{N}(0, \tF \tOmega \tF^\top +D)$ belongs to the family of \textit{extended} risk models \eqref{eq:family2}.
Set the family of all components to be estimated as $\theta=:(F_2, \Omega, D)$. We estimate $\theta$ by solving
\begin{equation}
    \label{eq:main_log}
    \maximize_\theta \sum_{\tau \leq T} w_\tau \ell(\theta; x_{\tau, \mathrm{obs}}),
\end{equation}
where
\[
\ell(\theta; x_{\tau, \mathrm{obs}}) = \log p_{\theta}(x_{\tau, \mathrm{obs}}),
\]
and, with a slight abuse of notation, $p_{\theta}(x_{\tau, \mathrm{obs}})$ is the marginal density of the observed returns at time $\tau$ induced by the density function $p_\theta(x)$.
% \ejc{$p_{\theta, \mathrm{obs}, \tau}(\cdot)$ is an ugly piece of notation, I would write it as $p_\theta(x_{\tau,\mathrm{obs}})$ and warn the reader that we are taking a shortcut.}
The weights $w_\tau$ are non-negative and sum to $1$. In our setting, they reflect the assumption that more recent observations are more informative and should therefore receive larger weights. Although we do not write it explicitly, $\Omega$ is positive definite and $D$ is diagonal with positive entries on its diagonal.

Problem \eqref{eq:main_log} maximizes the weighted sum of the Gaussian log-likelihoods of the observed data.
Although the distribution of financial returns may not be Gaussian,
% maximizing the Gaussian log likelihood is reasonable because financial returns are approximately Gaussian in the bulk of their distribution \cite{johansson2023simple}.
a Gaussian model is still a reasonable second-order approximation to the cross-sectional structure; in many risk-model applications the main object of interest is the covariance matrix, not the exact tail behavior \cite{johansson2023simple}. The Gaussian assumption also gives a simple tractable framework for estimating and manipulating covariance structure.
% \ejc{In the end, it's just another loss and is as defensible as any other loss.}
% \ejc{Not sure about this justification. I will need to take a pass. (fixed)}

For most daily or even weekly equity returns, the means are typically tiny compared to the random fluctuations. Therefore, assuming a zero mean often has little practical effect when the goal is to estimate the covariance. Nevertheless, we can account for a non-zero mean, by first estimating the mean return and then subtracting it from the return samples.

% \ejc{The rest of this section says little and I would delete it.}

% Assuming that there are no missing entries in the sequence of returns $x_1, \dots, x_t$, then problem \eqref{eq:main_log} is equivalent to
% \begin{equation*}
%     \minimize_\theta \mathrm{KL}\left(\mathcal{N}(0, C) \Vert \mathcal{N}(0, \tF \tOmega \tF^\top +D)\right),
% \end{equation*}
% where $\mathrm{KL}(\cdot \Vert \cdot)$ denotes the Kullback–Leibler divergence and
% \begin{equation}
% \label{eq:Cbad}
%     C = \sum_{\tau \leq t} w_\tau x_\tau x_\tau^\top.
% \end{equation}
% The proof is given in \Cref{appendix:equivalence}. This implies that our method can be thought of as maximizing log-likelihood or equivalently as minimizing the KL divergence with a time weighted empirical covariance matrix.

% As an aside, we should note that the minimizer of the Gaussian KL objective is equivariant under congruence transformations. Therefore, although not a focus of this paper, we could equivalently work with standardized returns.

% In \eqref{eq:C} we assume the sample returns $x_1, \dots, x_t$ are zero-mean. \rev{If we believe this is not the case, then we can first estimate the mean and then apply our method to the demeaned samples.}

\section{Solution via Expectation-Maximization}\label{sec:proposed}
We propose to solve problem \eqref{eq:main_log} using an iterative method called expectation-maximization (EM) \cite{hastie2009elements}. We first discuss the method in general. Then, we show how it applies to the case of fully observed returns. Finally, we discuss how to solve the problem with missing returns.

\paragraph{General Description} Consider maximizing
the objective
\begin{equation*}
\ell(\theta) = \sum_{\tau \le T} w_\tau \ell(\theta; y_\tau).
\end{equation*}
We have already seen this objective in the previous section in which $y_\tau$ was $x_{\tau,\text{obs}}$. We prefer here to use $y_\tau$ for generality.
In many problems, the log-likelihood is either intractable or leads to difficult optimization. In some cases, e.g., in our setting, augmenting the data with unobserved latent observations simplifies the log-likelihood and permits maximization. Consider then unobserved latent variables $z_\tau$ and the weighted sum of log-likelihoods of the augmented data:
\[
\ell^{\text{aug}}(\theta) = \sum_{\tau \le T} w_\tau \ell(\theta; y_\tau, z_\tau),
\]
where
\[
\ell(\theta; y_\tau, z_\tau) = \log p_\theta(y_\tau, z_\tau),
\]
and $p_\theta(y_\tau) = \int_{z_\tau} p_\theta(y_\tau, z_\tau) dz_\tau$.
Starting from an initial guess $\theta_0$, the EM algorithm inductively defines
\[
Q(\theta \, || \, \theta_k) = \mathbb{E} [\ell^{\text{aug}}(\theta) \mid \{y_\tau\}_{\tau = 1, \ldots, T}, \theta_k],
\]
where the conditional expectation above is over the latent variables $z_\tau \overset{\text{indep.}}{\sim}  p_{\theta_k}(z_\tau | y_\tau)$. This is the E-step. Then, the M-step finds
\[
\theta_{k+1} = \text{arg max}_\theta \,\, Q(\theta \, || \, \theta_k).
\]
The point of the data augmentation is that the M-step is now tractable.
It is well known that EM produces a monotone sequence in the sense that $\ell(\theta_k)$ is nondecreasing.

\paragraph{Fully Observed Returns} In the case where asset returns are fully observed, the observations $y_\tau$ above are the asset returns $x_\tau$. The latent variables are the factor returns so that with the notation above $z_\tau = s_\tau$.
The joint distribution $p_\theta(x_\tau, s_\tau)$ of asset returns and factor returns is
\[
s_\tau \sim \mathcal{N}(0, \tilde \Omega) \quad \text{and} \quad x_\tau | s_\tau \sim \mathcal{N}(\tilde F s_\tau, D).
\]
Marginally $x_\tau \sim \mathcal{N}(0, \tF \tOmega \tF^\top +D)$, as is requested. The E-step and M-step have closed-form expressions and their derivations are in \Cref{appendix:em}. With these expressions, the steps of the EM algorithm are given in \Cref{alg:em}. \rev{In the case that the number of factors in $F_1$ is zero ($n_1=0$), our algorithm reduces to the standard EM algorithm \cite{rubin1982algorithms} for factor models.} \Cref{alg:em} typically converges after $20-30$ iterations. That said, we do wish to clarify some of the steps in \Cref{alg:em} below.

% \ejc{I don't find this paragraph especially enlightening. (fixed).}
At iteration $k$, the parameter estimate is $\theta_k$. It holds that
\[
s_\tau \mid x_\tau; \theta_k \sim \mathcal{N}(L_k x_\tau, G_k).
\]
$L_k$ and $G_k$ are computed in lines 14-15 using standard multivariate Gaussian calculations. Furthermore,
\[
C_{ss}^k = \sum_{\tau \leq T} w_\tau \mathbb{E} \left[ s_\tau s_\tau^\top \mid x_\tau, \theta_k \right],
\]
\[
C_{xs}^k = \sum_{\tau \leq T} w_\tau \mathbb{E} \left[ x_\tau s_\tau^\top \mid x_\tau, \theta_k \right].
\]
$C_{ss}^k$ is computed in line 15, while $C_{xs}^k$ is computed in line 16.
Finally, in lines 17-21 we calculate $\theta_{k+1}$.

Particular attention should be given to the initialization $\theta_0=(F_2^0, \Omega_0, D_0)$ in \Cref{alg:em}. Once again, we wish to build the approximation
\begin{equation}
\label{eq:C}
    C = \sum_{\tau \le T} w_\tau x_\tau x_\tau^\top \, \approx \, \tilde F \tilde \Omega \tilde F^\top + D.
\end{equation}
The initialization for $F_2$ is based on the singular value decomposition of the residual returns of the base model. This is because the additional factors should explain the variance left unexplained by the base model. The initialization for $D$ preserves the asset volatilities, as given by the empirical covariance $C$. This means that the diagonal entries of the left-hand and right-hand sides in the above display \eqref{eq:C} are the same.

\paragraph{Partially Observed Returns} We let $y_\tau$ be the observed returns $x_{\tau, \mathrm{obs}}$ and $z_\tau$ be the factor returns and missing asset returns $(s_\tau, x_{\tau, \mathrm{mis}})$. The joint distribution of $(x_{\tau,\text{obs}}, x_{\tau,\text{mis}}, s_\tau)$ is $p_\theta(x_\tau, s_\tau)$ (up to a permutation of the variables).  The E-step is a little more involved but still has a closed-form expression.
The M-step also has a closed-form expression. The complete analysis is given in \Cref{appendix:em_mis}.

\begin{algorithm}[h!]
\caption{Extending a factor model using EM}\label{alg:em}
\begin{algorithmic}[1]
\State \textbf{Inputs:} sample returns $x_1, \dots, x_T \in \mathbb{R}^n$, base risk model $F_1 \Omega_\mathrm{base} F_1^\top + D_\mathrm{base}$ where $F_1 \in \mathbb{R}^{n \times n_1}$
\State \textbf{Parameters:} weights $w_1, \dots, w_T$ that are non-negative and sum to $1$, number of additional factors $n_2$, number of iterations $N_\mathrm{iterations}$, small positive threshold $\nu >0$
\State Obtain $C$ as in \eqref{eq:C}.

\Comment{Initialization}
\State Regress the returns $x_1, \dots, x_T$ on the factor exposure matrix $F_1$ and obtain the residuals $e_1, \dots, e_T \in \mathbb{R}^n$.

\State Compute the empirical covariance $\dfrac{1}{T} \sum_{\tau \le T} e_\tau e_\tau^\top$. Take its eigendecomposition. Let $\Lambda \in \mathbb{R}^{n_2 \times n_2}$ be the diagonal matrix of the top $n_2$ eigenvalues of the covariance and $U \in \mathbb{R}^{n \times n_2}$ be the matrix of the corresponding eigenvectors.
\State $k \leftarrow 1$
\State $F_2^k \leftarrow U \Lambda^{1/2}$
\State $\Tilde{F}_k \leftarrow \begin{bmatrix}
    F_1 & F_2^k
\end{bmatrix}$
\State $\Omega_k \leftarrow \Omega_\mathrm{base}$
\State $\Tilde{\Omega}_k = \begin{bmatrix}
    \Omega_k & 0\\
    0 & I
\end{bmatrix}$
\State $D_k = \mathrm{max}\lbrace \mathrm{diag}\left(C - \Tilde{F}_k \Tilde{\Omega}_k \Tilde{F}_k^\top\right), \nu\rbrace$, where $\mathrm{max}$ is elementwise and $\mathrm{diag}(\cdot)$ with a matrix argument returns the diagonal elements.
\For{$ k= 1, \dots, N_\mathrm{iterations}$}

\Comment{Expectation-step}
    \State $G_k^{-1} \leftarrow \tF_k^\top D_k^{-1} \tF_k + \tOmega_k^{-1}$
    \State $L_k \leftarrow G_k \tF_k^\top D_k^{-1}$, where $L_k = \begin{bmatrix}
        L_1^k \\
        L_2^k
    \end{bmatrix}$, where $L_1^k \in \mathbb{R}^{n_1 \times n}$, $L_2^k \in \mathbb{R}^{n_2 \times n}$.
    \State $C_{ss}^k \leftarrow G_k + L_k C L_k^\top$
    where $C_{ss}^k =\begin{bmatrix}
        C_{ss, 11}^k & C_{ss, 12}^k \\
        C_{ss, 12}^{k,\top} & C_{ss, 22}^k
    \end{bmatrix}$ and $C_{ss, 11}^k \in \mathbb{R}^{n_1 \times n_1}$, $C_{ss, 22}^k \in \mathbb{R}^{n_2 \times n_2}$.
    \State $C_{xs}^k \leftarrow C L_k^\top$, where $C_{xs}^k = \begin{bmatrix}
        C L_{1}^{k, \top} & C L_{2}^{k, \top}
    \end{bmatrix} = \begin{bmatrix}
        C_{xs, 1}^{k} & C_{xs, 2}^{k}
    \end{bmatrix}$

    \Comment{Maximization-step}
    \State $\Omega_k \leftarrow C_{ss, 11}^k$
    \State $\Tilde{\Omega}_k = \begin{bmatrix}
    \Omega_k & 0\\
    0 & I
\end{bmatrix}$
    \State $F_2^k \leftarrow \left( C_{xs, 2}^k - F_1 C_{ss,12}^k \right)  C_{ss, 22}^{k, -1}$
    \State $\Tilde{F}_k \leftarrow \begin{bmatrix} F_1 & F_2^k \end{bmatrix}$
    \State $D_k \leftarrow \diag \left(C - 2 C_{xs}^k \Tilde{F}_k^\top + \Tilde{F}_k C_{ss}^k \Tilde{F}_k^\top \right)$
\EndFor
\State \Return $F_2^k, \Omega_k, D_k$
\end{algorithmic}
\end{algorithm}

\section{Empirical Performance}
\label{sec:resultss}

We consider the time period between 2018-06-27 and 2023-12-28. This period contains the COVID-19 pandemic, a time of high market volatility. Our universe consists of the US large-capitalization equities in the BlackRock Systematic investment universe for which we observe complete return data over the period from 2018-06-27 to 2023-12-28.
We restrict attention to assets with uninterrupted return histories and for simplicity do not consider the version of our approach that explicitly accommodates missing returns in evaluating performance. As a result, the reported performance may not be representative of the performance on the full evolving asset universe. Nevertheless, it allows us to focus on the core mechanism by which our approach refines a general risk model to a targeted investment universe. This restriction does not necessarily favor the extended model over the base model.
% \rf{Since our analysis concerns statistical fit rather than implementable trading strategies, the resulting selection effect is less consequential than it would be in a backtest of realized portfolio performance.}
% \ejc{Why is it a problem for a variance estimation problem? (fixed)}

Our universe consists of 870 assets for which we have access to daily returns. \rf{Every day $t$, we can observe the history of daily returns $x_1, \dots, x_t$}. The return $x_\tau$ corresponds to the time period between day $\tau-1$ and day $\tau$.  The low-rank-plus-diagonal covariance estimate at time $t$ is denoted
\begin{equation*}
    \hat{\Sigma}_t = \hat{F}_t \hat{F}_t^\top + \hat{D}_t.
\end{equation*}
Note that we hide the factor return covariance (specifically its square root) in $\hat{F}_t$. Below, $\hat{\Sigma}_t$ will correspond to either:
\begin{itemize}
    \item the base model at time $t$, in which case $\hat{F}_t = F_{1,t} \Omega_{t, \mathrm{base}}^{1/2}$ ($F_{1, t}$ is the provided factor exposure matrix at time $t$) and $\hat{D}_t = D_{t, \mathrm{base}}$ ($D_{t, \mathrm{base}}$ is the provided diagonal matrix of idiosyncratic variances at time $t$).
    \item the extended model at time $t$, in which case $\hat{F}_t = \begin{bmatrix}
        F_{1,t} & \hat{F}_{2,t}
    \end{bmatrix} \begin{bmatrix}
        \hat{\Omega}_t^{1/2} & 0 \\
        0 & I
    \end{bmatrix}$, where $\hat{F}_{2,t}$, $\hat{\Omega}_t$, and $\hat{D}_t$ are obtained by \Cref{alg:em} using the observed returns at time $t$, i.e., $x_1, \dots, x_t$.
    % \item \ejc{Not entirely convinced about this.} an \textit{exposure-misaligned} extended risk model at time $t$, in which case $\hat{F}_t = \begin{bmatrix}
    %     F_{1,t} & \Pi \hat{F}_{2,t}
    % \end{bmatrix} \begin{bmatrix}
    %     \hat{\Omega}_t^{1/2} & 0 \\
    %     0 & I
    % \end{bmatrix}$, where $\hat{F}_{2,t}$, $\hat{\Omega}_t$, and $\hat{D}_t$ are obtained by \Cref{alg:em} and $\Pi$ is a random permutation matrix. This model serves as a baseline to show that uncovering hidden structure in the asset returns is not a trivial problem, i.e., it is a control or safety guard. $\Pi \hat{F}_{2,t}$ assigns the exposures to the additional factors of one asset to another asset, destroying any meaningful structure that was learned by \Cref{alg:em}. We call this baseline the \textit{extended (permuted)} model below.
\end{itemize}

% MK: is the 53 proprietary? If so, maybe say "several dozen" instead of 53? Good question for Ron.
Our base model is the Barra short-term US risk model, which is updated monthly\footnote{The data that support the findings of this study are not publicly available because they are proprietary and subject to confidentiality and data-use restrictions. The authors are not authorized to redistribute the raw data or derived datasets. Interested researchers may request access from the relevant data provider, subject to approval and any applicable licensing, confidentiality, and data-use agreements. The results reported in this article were produced using the minimum dataset necessary to conduct the analyses described.}. The base model consists of $73$ factors for the universe of 870 assets. The $73$ factors include sectors, such as aerodefense, airlines, and banks, where the exposure is binary ($0$ or $1$), and factors where the exposures are not binary, e.g., the beta factor. We extend the base model with an additional $7$ factors. This number is considered a good choice overall, as it allows the extended model to learn information missed by the base model, but also mitigates the risk of overfitting. We update the extended model on a daily basis. We use an EWMA with a half-life of $H=126$ days for the weights $w_\tau$ on each day $t$.
Our ability to pick up transient factors depends on the choice of the half-life parameter. A short half-life will allow us to pick up short-lived, rapidly changing effects, but if it is too short it may primarily capture noise rather than meaningful signal. A longer half-life will emphasize more persistent, slowly evolving \rf{correlations}, but may \rf{be biased}.
We do not optimize over the choice of the half-life or the number of additional factors in this paper.

We compare the two models--the base and the extended models--with respect to their statistical fit out-of-sample. For this purpose we consider
several metrics that correspond to separate subsections below.
We evaluate the models in the time period between 2019-06-26 and 2023-12-28. The days before 2019-06-26 are used as the burn-in period for the extended model.

Throughout, we assume that daily returns are zero-mean. \rf{This approximation is reasonable for most stocks, because the mean returns are small compared to the random fluctuations. For example, the S\&P 500 has a daily mean return of 4 basis points (bps) and a daily volatility of 100 bps.}

\subsection{Out-of-Sample $R^2$}

% If the extended model uncovers covariance structure that the base model misses, then it should be able to better estimate out-of-sample returns.
% We quantify this in \Cref{sec:r2_ret} by measuring the out-of-sample $R^2$. A larger $R^2$ implies a better statistical fit. A negative $R^2$ implies that the model is worse than simply predicting zero.

% If the extended model uncovers covariance structure that the base model misses, then it should be able to better estimate out-of-sample returns.
% % In other words, the extended model
% % should explain more of the out-of-sample returns variance.
% We measure this by computing the time and asset out-of-sample $R^2$. \ejc{What does the last sentence mean?}
% We analyze the $R^2$ obtained from estimating returns and the $R^2$ obtained from estimating the residuals relative to the base factors, recognizing that these quantities are closely related. A higher $R^2$ implies a better statistical fit. A negative $R^2$ implies that the model is worse than simply predicting zero. \ejc{This paragraph is not clear. I can't understand it even though I think I proposed the method.}

\subsubsection{Estimating Returns}\label{sec:r2_ret}

If the extended model uncovers covariance structure that the base model misses, then it should be able to better estimate out-of-sample returns.
We quantify this by measuring the out-of-sample $R^2$. A larger $R^2$ implies a better statistical fit. A negative $R^2$ implies that the model is worse than simply predicting zero.

At time $t$, we split the $n$ assets into two disjoint groups: $\mathrm{train}$ and $\mathrm{test}$ (we suppress the dependence on $t$ for these two groups for clarity).
We use the next-day returns for the assets in the train set to get the most likely factor returns according to our model, by solving
% We solve the out-of-sample maximum likelihood problem tied with the risk model \eqref{eq:risk_notation}, but limited
% to the set of assets $G_\mathrm{train}$:
\begin{equation}\label{eq:ltsq_r^2}
    \minimize_{s_t} \left(x_{t+1}^\mathrm{train} - \hat{F}_t^\mathrm{train} s_t \right)^\top \hat{D}_t^{\mathrm{train}, -1} \left( x_{t+1}^\mathrm{train} - \hat{F}_t^\mathrm{train} s_t \right) + s_t^\top s_t,
    % \lVert r_{t+1} - \hat{F}_t s_t \rVert_2^2
\end{equation}
where the superscript $\mathrm{train}$ indicates that we only consider the assets in the $\mathrm{train}$ set. We maximize the conditional density of the factor returns given the train returns.
% \rev{In other words, our estimate \rf{$s_t^\star$} of the factor returns is the maximum likelihood estimate under our factor model.}
We then predict the next-day (out-of-sample) returns for the assets in the $\mathrm{test}$ set as
\begin{equation*}
    \hat{x}_{t+1}^\mathrm{test} = \hat{F}_t^\mathrm{test} \hat{s}_t,
\end{equation*}
where $\hat{s}_t$ is the solution of \eqref{eq:ltsq_r^2}.  \rev{This is the return prediction under our factor model and it is equal to the 
conditional expectation 
\[
\mathbb{E}\left[ x_{t+1}^\mathrm{test} \mid x_{t+1}^\mathrm{train} \right]
\]
in the model
\[
x_{t+1} = \hat{F}_t s + \varepsilon, \quad s\sim \mathcal{N}(0, I), \quad \varepsilon \sim \mathcal{N}(0, \hat{D}_t), \quad s \perp \varepsilon.
\]} 
The conditional expectation is the best predictor in the mean-squared error sense. 
Finally, we compute the $R^2$ between the predicted next-day returns and the actual next-day returns for the assets in the $\mathrm{test}$ set:
\begin{equation*}
    R^2_t = 1 - \dfrac{\lVert x_{t+1}^\mathrm{test} - \hat{x}_{t+1}^\mathrm{test} \rVert_2^2}{\lVert x_{t+1}^\mathrm{test} - 0\rVert_2^2}.
\end{equation*}
% where
% \begin{equation}
%     \mathrm{Avg}(x_{t+1}^\mathrm{test} )  = \dfrac{1}{\lvert G_\mathrm{test} \rvert} \sum_{i \in G_\mathrm{test}} x_{t+1}^i,
% \end{equation}
% and $x_{t+1}^i$ is the return of asset $i$ at time $t+1$.
\rev{Using zero as the mean is reasonable for daily asset returns \rf{because the mean is typically much smaller than the random fluctuations}.}
If the extended risk model achieves a higher $R^2$ than the base model, this suggests that it exploits missing correlations to make better predictions from available returns.
% A higher $R^2$ implies that the risk model captures more of the structure present in the returns, as it is able to better estimate out-of-sample returns.
We compute $R^2$ every day in the evaluation period.

% For each of $30$ replications, at each time $t$ we randomly split the assets into a test and train set ($90\%$ of the assets are in the train set and the remaining $10\%$ is in the test set at each time) and compute the corresponding out-of-sample return $R^2_t$. We then average $R^2_t$ across replications at each $t$ to obtain a cross-replication mean time series.

We include a third risk model in the results of this section as a baseline. We call this model \textit{randomly extended}. We construct this risk model as follows. Let $F_2' \in \mathbb{R}^{n \times 7}$ be a matrix of factor exposures sampled from $\mathcal{N}(0,1)$. This risk model is then
\[
\begin{bmatrix}
    F_{1, t} & F_2'
\end{bmatrix} \Omega_t'
\begin{bmatrix}
    F_{1, t} & F_2'
\end{bmatrix}^\top + D_t',
\]
where we learn $\Omega_t'$ and $D_t'$ using \Cref{alg:em}: we run \Cref{alg:em} using $n_2 = 0$ with the input $F_1$ being $\begin{bmatrix}
    F_{1, t} & F_2'
\end{bmatrix}$. Essentially, we augment the base model with a random factor exposure matrix and re-learn the factor return covariance and the idiosyncratic variances to maximize the weighted log-likelihood.
The added factor subspace for this model is random and for the small number of added factors we consider it is not expected to capture meaningful structure in the returns.

The time series of $R_t^2$ is shown in \Cref{fig:return_r2}. The extended model outperforms the base model. It is also important to note that the extended model, where the 7 additional factors are learned from the returns data, outperforms the randomly extended model. We therefore conclude that the extended model actually captures covariance structure that the base model misses. The randomly extended model slightly underperforms the base model because it overfits to the assets in the train set due to the random added factors. This emphasizes what should be obvious: adding random factors does not help in estimating out-of-sample returns. We do not show the randomly extended model in future sections.

Finally, we note that the base model also captures a significant amount of structure in the returns. The $90/10$ split is chosen to ensure that the train set is \rf{significantly large} to obtain meaningful factor return estimates, while the test set is large enough to obtain meaningful out-of-sample $R^2$ estimates. We have verified that other splits (e.g., $50/50$) lead to similar conclusions.

% The extended (permuted) model performs significantly worse: assigning the added factors' exposures of one asset to another washes out the covariance structure uncovered by \Cref{alg:em}. This indicates that uncovering hidden structure in the asset returns is not a trivial problem. The extended (permuted) model is significantly worse at the end of 2023. This is consistent with a significant regime shift that occurred in the US equity market during the final quarter of 2023. The transition was from a narrow market led only by a few technology giants to a more broad-based rally. In this period, misaligning exposures across assets leads to a significant degradation in performance.

For each risk model we report a point estimate and a dispersion for the $R^2$ in \Cref{tab:return_r2}.
% The reported point estimate corresponds to the time average of the cross-replication mean. The reported dispersion measure is the standard deviation of the average across time of the cross-replication mean. These results are included in \Cref{tab:return_r2}.
The extended model offers an improvement in $R^2$. It allows for better return forecasts of $10\%$ of the returns given the remaining $90\%$ of the returns.

\begin{figure}[]
    \centering
    \includegraphics[width=0.7\linewidth]{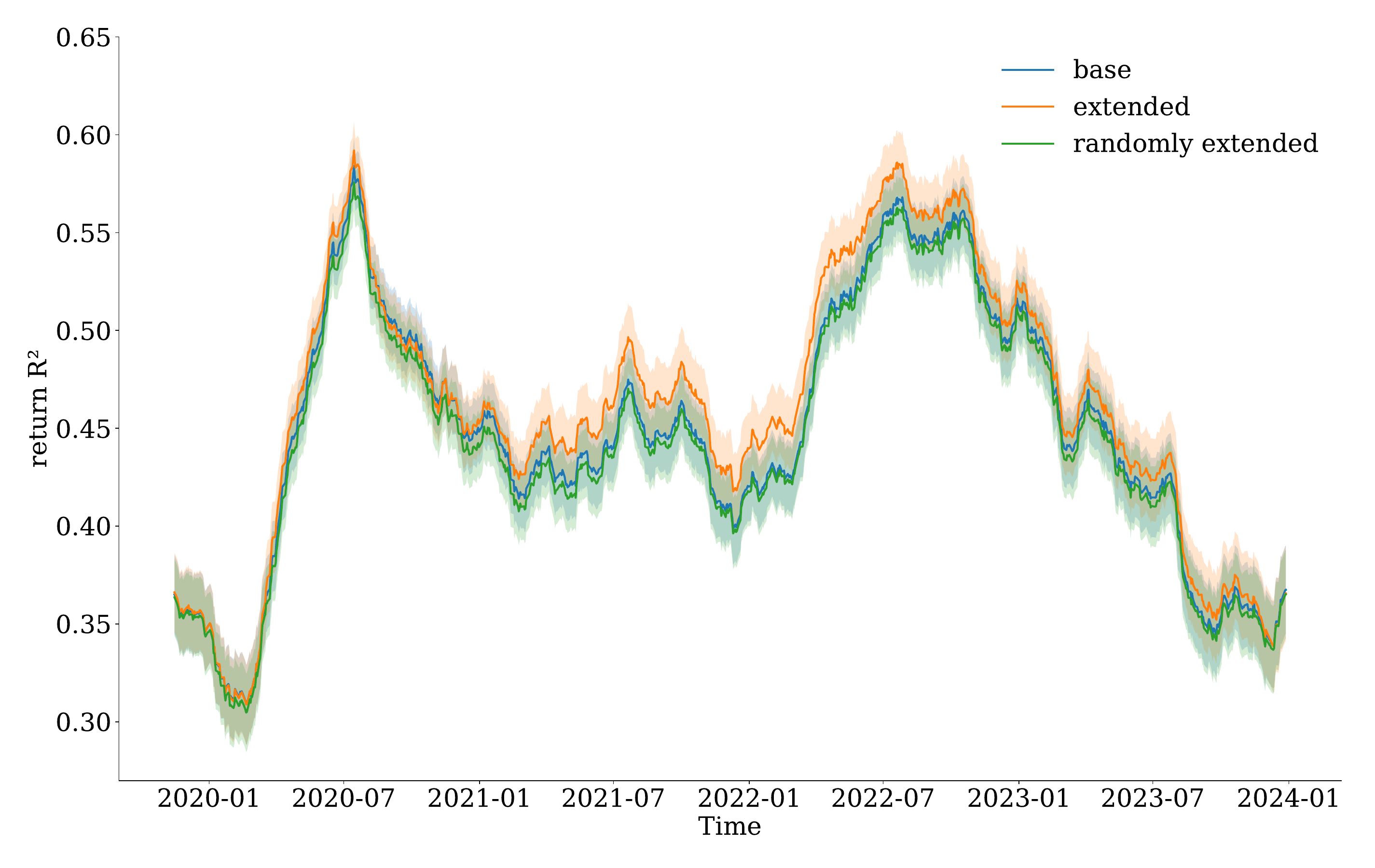}
    \caption{Out-of-sample return predictability ($R^2$) across time. At each of $30$ replications and each day, assets are randomly split into train and test sets (90/10), and the out-of-sample (next-day for test assets) return $R^2$ is computed. Lines report the rolling mean (window = 100 days) of the cross-replication average $R^2$. The shaded bands indicate the $\pm 1$ rolling mean of the standard deviation of the average cross-replication $R^2$ (window = 100 days), providing a smoothed measure of sampling variability across random splits.}
    \label{fig:return_r2}
\end{figure}

\begin{table}[ht]
\centering
\begin{tabular}{lcc}
\toprule
Model
& \makecell{Point estimate} &
\makecell{Dispersion} \\
\midrule
Base     & 0.445  & 0.0058 \\
Extended & \textbf{0.454} & 0.0059 \\
Randomly extended & 0.439 & 0.0058 \\
\bottomrule
\end{tabular}
\caption{Results for the out-of-sample return $R^2$ using a $90/10$ train/test split. The point estimate is the average across time of the cross-replication mean. The dispersion is the standard deviation of the average across time of the cross-replication mean.}
\label{tab:return_r2}
\end{table}

\subsubsection{Estimating Residuals with Respect to the Factors of the Base Model}\label{sec:r2_res}
The extended model assumes that there is additional structure in the returns beyond the factors included in the base model. If this additional structure exists and the extended model is able to learn it, then it should be able to predict out-of-sample residuals with respect to the factors of the base model using the additional learned factors.
This methodology is similar, but not identical, to the measure of out-of-sample performance proposed by Spector et al. \cite{spector2024mosaic}.
% We repeat the same procedure as above, but now we compute $R^2$ for the estimation of the residuals with respect to the base factors \ejc{What's this?}, which we also call the \textit{base} residuals. We check whether the extended model is able to capture structure in the out-of-sample base residuals that is missed by the base model. This methodology is similar, but not identical, to the measure of out-of-sample performance proposed by Spector et al. \cite{spector2024mosaic}.

We only consider the extended model in this section. Recall that in this case
$\hat{F}_t = \begin{bmatrix}
        \hat{F}_{1,t} & \hat{F}_{2,t}
    \end{bmatrix}$, where $\hat{F}_{1,t}=F_{1, t} \hat{\Omega}_t^{1/2} $ corresponds to the $n_1$ factors also present in the base model and $\hat{F}_{2,t}$ corresponds to the additional $n_2$ factors.
The underlying return model is therefore
\[
x_t = \hat{F}_{1,t} s_{1,t} + \hat{F}_{2,t} s_{2,t} + \hat{D}_t^{1/2} \epsilon_t,
\]
where $s_{1,t}$, $s_{2,t}$, and $\epsilon_t$ are all decorrelated.
We refer to the $n_1$ factors involved in $\hat{F}_{1,t}$, $s_{1,t}$, as the \textit{base factors}.

At each time $t$, we split the $n$ assets into two disjoint groups: $\mathrm{train}$ and $\mathrm{test}$ (again we suppress the dependence on $t$).
We use the next-day returns for the assets in the train set to compute the factor returns for the $n_2$ added factors.
We then use these returns to estimate the residuals with respect to the base factors for the assets in the test set.
% \ejc{I feel presentation is a little complicated here. Is there a way to simplify it? (fixed)}

In detail, we first use the assets in the train set to get the most likely base factor returns, by solving
\begin{equation}\label{eq:ltsq_r^2_residuals}
    \begin{aligned}
\minimize_{s_{1,t}} &\left(x_{t+1}^\mathrm{train} - \hat{F}_{1,t}^{\mathrm{train}} s_{1,t} \right)^\top \left( \hat{F}_{2,t}^{\mathrm{train}} \hat{F}_{2,t}^{ \mathrm{train}, \top } + \hat{D}_t^{\mathrm{train}} \right)^{-1} \\
&\qquad \left( x_{t+1}^\mathrm{train} - \hat{F}_{1,t}^{\mathrm{train}} s_{1,t}  \right) + s_{1,t}^{\top} s_{1,t}.
% \lVert r_{t+1} - \hat{F}_t s_t \rVert_2^2
    \end{aligned}
\end{equation}
% where the superscript $\mathrm{train}$ indicates that we only consider the assets in $G_\mathrm{train}$.
The residuals with respect to the base factors for the train assets are then:
% The $G_\mathrm{train}$ residuals with respect to the base factors are then:
\begin{equation*}
    \epsilon_{t+1}^\mathrm{train} = x_{t+1}^\mathrm{train} - \hat{F}_{1,t}^{ \mathrm{train}} {\hat{s}_{1,t}^\mathrm{train}},
\end{equation*}
where $\hat{s}_{1,t}^\mathrm{train}$ is the solution of \eqref{eq:ltsq_r^2_residuals}.
Problem \eqref{eq:ltsq_r^2_residuals} is appropriate because \rf{it maximizes the conditional density of the base factor returns, $s_{1,t}$, given the returns in train}.

We then use the train assets to obtain the most likely returns for the $n_2$ added factors by solving
\begin{equation}
\label{eq:ltsq_add}
    \minimize_{s_{2,t}} \left(\epsilon_{t+1}^\mathrm{train} - \hat{F}_{2,t}^{\mathrm{train}} s_{2,t} \right)^\top \hat{D}_t^{\mathrm{train}, -1} \left( \epsilon_{t+1}^\mathrm{train} - \hat{F}_{2,t}^{ \mathrm{train}} s_{2,t} \right) + s_{2,t}^{\top} s_{2,t},
\end{equation}
with solution $\hat{s}_{2,t}$. Problem \eqref{eq:ltsq_add} maximizes the conditional density of the added factor returns, $s_{2,t}$, given that
\[
\hat{F}_{2,t}^\mathrm{train} s_{2,t} + \hat{D}_t^{\mathrm{train}, 1/2} \epsilon_t =  \epsilon_{t+1}^\mathrm{train}.
\]

Similarly, we compute the residuals with respect to the base factors for the test assets by solving
\begin{equation}\label{eq:ltsq_r^2_residuals_test}
    \begin{aligned}
\minimize_{s_{1,t}} &\left(x_{t+1}^\mathrm{test} - \hat{F}_{1,t}^{ \mathrm{test}} s_{1,t} \right)^\top \left( \hat{F}_{2,t}^{\mathrm{test}} \hat{F}_{2,t}^{\mathrm{test}, \top }+\hat{D}_t^{\mathrm{test}} \right)^{-1} \left( x_{t+1}^\mathrm{test} - \hat{F}_{1,t}^{\mathrm{test}} s_{1,t}  \right) \\
&+ s_{1,t}^{ \top} s_{1,t} ,
% \lVert r_{t+1} - \hat{F}_t s_t \rVert_2^2
    \end{aligned}
\end{equation}
and setting
\begin{equation*}
    \epsilon_{t+1}^\mathrm{test} = x_{t+1}^\mathrm{test} - \hat{F}_{1,t}^{\mathrm{test}} {\hat{s}_{1,t}^\mathrm{test}},
\end{equation*}
where ${\hat{s}_{1,t}^\mathrm{test}}$ is the solution of \eqref{eq:ltsq_r^2_residuals_test}.

Finally, we compute the predicted residuals with respect to the base factors for the test assets as
\begin{equation*}
    \hat{\epsilon}_{t+1}^\mathrm{test} = \hat{F}_{2,t}^{\mathrm{test}} {\hat{s}_{2,t}},
\end{equation*}
and obtain the $R^2$
\begin{equation*}
    R^2_t = 1 - \dfrac{\lVert \epsilon_{t+1}^\mathrm{test} - \hat{\epsilon}_{t+1}^\mathrm{test} \rVert_2^2}{\lVert \epsilon_{t+1}^\mathrm{test} -  {0}\rVert_2^2}.
\end{equation*}
% where
% \begin{equation}
%     \mathrm{Avg}(\epsilon_{t+1}^\mathrm{test} )  = \dfrac{1}{\lvert G_\mathrm{test} \rvert} \sum_{i \in G_\mathrm{test}} \epsilon_{t+1}^i.
% \end{equation}
A positive $R^2$ here indicates that the added factors are able to explain existing structure in the out-of-sample residuals with respect to the base factors.
\rev{We note that $\hat{\epsilon}_{t+1}^\mathrm{test}$ only depends on the next day's returns of the train assets, while $\epsilon_{t+1}^\mathrm{test} $ only depends on the next day's returns of the test assets. Our ability to predict the latter using the former depends on the quality of the risk model.}

The $R^2$ is plotted in \Cref{fig:residual_r2}. The extended model achieves a positive $R^2$ which indicates that it is in fact learning useful structure in the returns. The added factors allow the model to predict out-of-sample residuals with respect to the base factors. The average across time cross-replication mean $R^2$ is $0.125$.

\begin{figure}[h]
    \centering
    \includegraphics[width=0.7\linewidth]{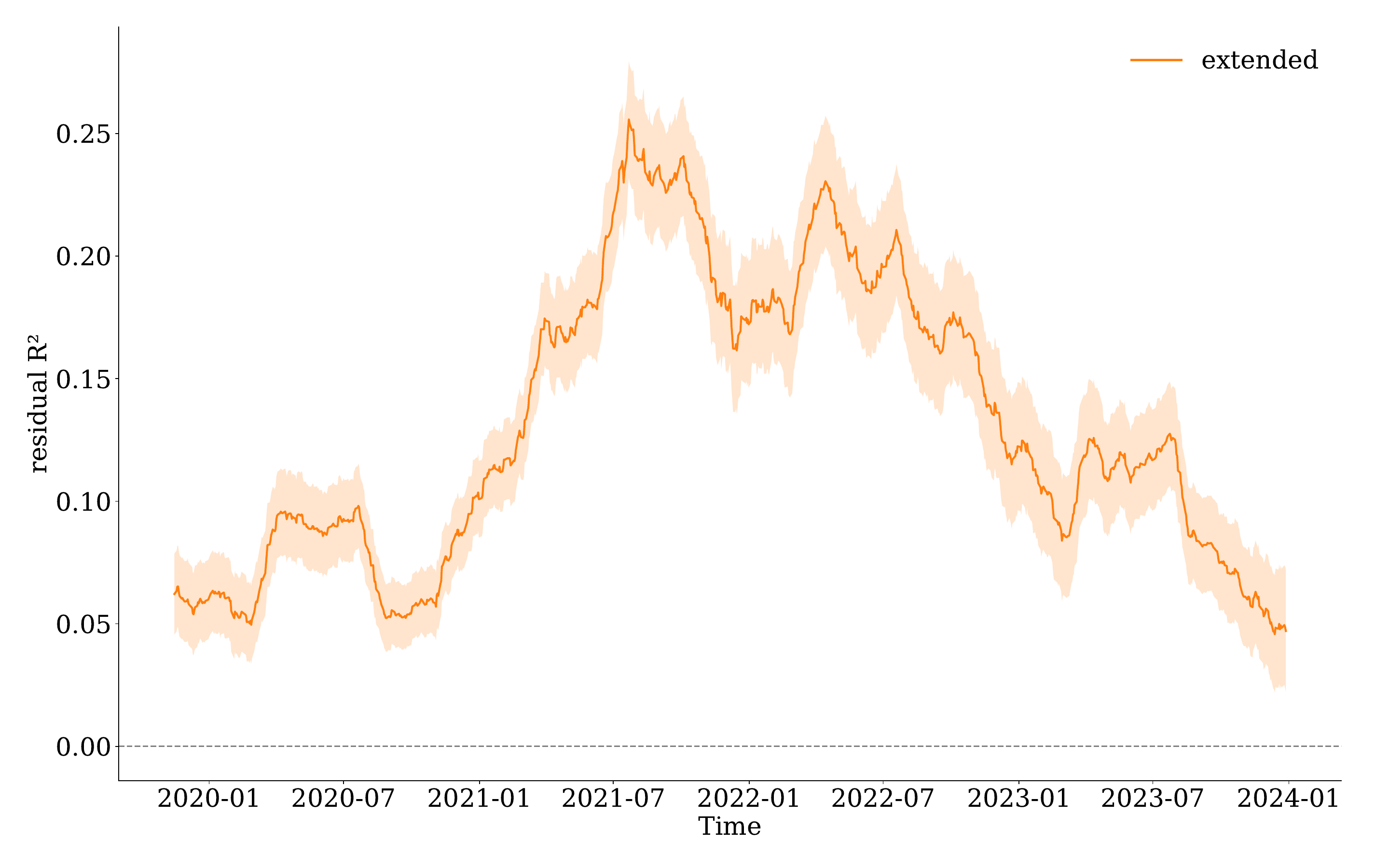}
    \caption{
    Out-of-sample predictability ($R^2$) of the residuals with respect to the base model's factors based on the added learned factors for the extended model across time. At each of $30$ replications and each day, assets are randomly split into train and test sets (90/10), and the out-of-sample (next-day for test assets) $R^2$ for the residuals is computed. Lines report the rolling mean (window = 100 days) of the cross-replication average $R^2$. The shaded bands indicate the $\pm 1$ rolling mean of the standard deviation of the average cross-replication $R^2$ (window = 100 days), providing a smoothed measure of sampling variability across random splits.}
    \label{fig:residual_r2}
\end{figure}

% \begin{table}[ht]
% \centering
% \begin{tabular}{lcc}
% \toprule
% Model
% & \makecell{Point estimate} &
% \makecell{Dispersion} \\
% \midrule
% Base     & 0.0  & 0.0 \\
% Extended & \textbf{0.125} & 0.005 \\
% Extended (permuted) & $-$0.013 & 0.001 \\
% \bottomrule
% \end{tabular}
% \caption{\rf{Results for the out-of-sample residual $R^2$ using a $90/10$ train/test split. The point estimate is the time average of the cross-replication mean. The dispersion is the standard deviation of the average across time of the cross-replication mean. The extended model allows for better residual forecasts of $10\%$ of the returns given the remaining $90\%$}.}
% \label{tab:return_r2}
% \end{table}

\subsection{Predictive Ability of Added Factors}

We demonstrate in another way that the added factors of the extended model
\[
\begin{bmatrix}
    \hat{F}_{1, t} & \hat{F}_{2, t}
\end{bmatrix}\begin{bmatrix}
    \hat{F}_{1, t} & \hat{F}_{2, t}
\end{bmatrix}^\top + \hat{D}_t
\]
capture structure present in the returns that the base factors miss.
% Similarly to the last section, we consider the extended risk model $\hat{F}_t = \begin{bmatrix}
%         \hat{F}_{1,t} & \hat{F}_{2,t}
%     \end{bmatrix}$, where $\hat{F}_{1,t}=F_{1, t} \hat{\Omega}_t^{1/2} $ corresponds to the $n_1$ base factors and $\hat{F}_{2,t}$ corresponds to the added $n_2$ factors.
We consider the null hypothesis
% \ejc{What's Gaussian here? I don't get it.}
\[
x_t = \hat{F}_t^{(1)} s_t^{(1)} +\Tilde{\epsilon}_t,\quad (H_0),
\]
where $\Tilde{\epsilon}_t$ and $s_t^{(1)}$ are independent, and the elements of $\Tilde{\epsilon}_t$ are independent. In other words, we assume there is no added structure, other than the base factors. We will show that our predictive ability is not consistent with this hypothesis.

% \ejc{We do not need to repeat all the equations. Just say
% \begin{equation*}
%     \minimize_{s_t} \lVert x_{t+1}^\mathrm{fold} - \hat{F}_t^{(1), \mathrm{fold}} s_t \rVert_2^2 .
%     % \lVert r_{t+1} - \hat{F}_t s_t \rVert_2^2
% \end{equation*}
% where fold is either rain or test. Also only label equations you recall. (fixed)}

We split the $n$ assets into two disjoint groups: $\mathrm{train}$ and $\mathrm{test}$ (again we suppress the dependence on $t$). For each group, we solve the out-of-sample ordinary least squares problem
\begin{equation}\label{eq:ltsq_ee}
    \minimize_{s_{1,t}} \lVert x_{t+1}^\mathrm{group} - \hat{F}_{1,t}^{\mathrm{group}} s_{1,t} \rVert_2^2 ,
    % \lVert r_{t+1} - \hat{F}_t s_t \rVert_2^2
\end{equation}
where $\mathrm{group}$ is either train or test. The solution of \eqref{eq:ltsq_ee} is $\hat{s}_{1,t}^{\mathrm{group}}$ for each group.
We compute the optimal residuals
\[
    \hat{\delta}_t^\mathrm{group} = x_{t+1}^\mathrm{group} - \hat{F}_{1,t}^{\mathrm{group}} {\hat{s}_{1,t}}^{\mathrm{group}}.
\]
Under the null $H_0$, it holds that $\hat{\delta}_t^\mathrm{train}$ and $\hat{\delta}_t^\mathrm{test}$ are independent.
\rf{In other words, under the null, the residuals $\hat{\delta}_t^\mathrm{train}$ are not predictive of the residuals $\hat{\delta}_t^\mathrm{test}$: the highest possible $R^2$ value is $0$.}
Now let
\[
    P_t^\mathrm{group} = \left( I - \hat{F}_{1,t}^{ \mathrm{group}} (\hat{F}_{1,t}^{\mathrm{group}, \top } \hat{F}_{1,t}^{\mathrm{group}})^{-1} \hat{F}_{1,t}^{\mathrm{group}, \top} \right) \hat{F}_{2,t}^{ \mathrm{group}}
\]
for each of the two groups: train or test.
Under the alternative hypothesis that the extended risk model holds, i.e., the additional factors exist and
\[
    x_t = \hat{F}_{1,t} s_{1,t} + \hat{F}_{2,t} s_{2, t} + \hat{D}_t^{1/2}\epsilon_t,\quad (H_1),
\]
the in-sample residuals of the train and test assets with respect to the base factors (i.e., after solving problems \eqref{eq:ltsq_ee} where we replace $x_{t+1}$ with $x_t$) have factor exposures to $s_{2,t}$ equal to $P_t^\mathrm{train}$ and $P_t^\mathrm{test}$, respectively. Therefore, to predict $s_{2,t}$ using $\hat{\delta}^\mathrm{train}$, we solve
\begin{equation*}
    \minimize_{s_{2,t}} \lVert \hat{\delta}_{t}^\mathrm{train} - P_t^\mathrm{train } s_{2,t} \rVert_2^2
    % \lVert r_{t+1} - \hat{F}_t s_t \rVert_2^2
\end{equation*}
and obtain the solution $\hat{s}_{2,t}$.
We then make the prediction
\[
    \Tilde{\delta}_t^\mathrm{test} = P_t^\mathrm{test} \hat{s}_{2,t}
\]
for $\hat{\delta}_t^\mathrm{test}$. Under the null, the best prediction would be $0$ for $\hat{\delta}_t^\mathrm{test}$ (assume $\hat{F}_{2,t}=0$) and the best $R^2$ would also be $0$. If the $R^2$
\begin{equation}
\label{eq:r2_res_exttt}
        R^2_t = 1 - \dfrac{\lVert \hat{\delta}_t^\mathrm{test} - \Tilde{\delta}_t^\mathrm{test}\rVert_2^2}{\lVert \hat{\delta}_t^\mathrm{test} -  0\rVert_2^2}
\end{equation}
is positive, then we are able to predict $\hat{\delta}_t^\mathrm{test}$ from $\hat{\delta}_t^\mathrm{train}$, which should not be possible under the null. Therefore, this would be evidence against the null and in favor of the discovered added factors. This methodology is inspired by Spector et al. \cite{spector2024mosaic}.

% \ejc{This section is also unnecessarily complicated. I would streamline. Also, comments from other sections apply. (fixed)}

The time series of $R^2$ is plotted in \Cref{fig:mosaic}.
The predictive ability implied by the positive out-of-sample $R^2$ for the extended model is evidence that the added factors are present (equivalently evidence against the null). The average across time of the cross-replication mean $R^2$ is $0.0129$. The time periods when the extended model achieves the largest return $R^2$ gap to the base model (as shown in \Cref{fig:return_r2}) coincide with the periods of largest $R^2$ in \Cref{fig:mosaic}. The performance degrades in late 2023. This is a period when the extended model does not improve upon the base model, as shown in \Cref{fig:return_r2}.
% The extended (permuted) model obtains an $R^2$ that is negative. This implies that a wrong exposure attribution for the added factors does not allow us to reject the null hypothesis of no added factors with this model.

% For each risk model we report a point estimate and a dispersion for the $R^2$. The reported point estimate corresponds to the time average of the cross-replication mean. The reported dispersion measure is the standard deviation of the average across time of the cross-replication mean. These results are included in \Cref{tab:return_r2}. The extended model offers an improvement in $R^2$.

\begin{figure}[H]
    \centering
    \includegraphics[width=0.7\linewidth]{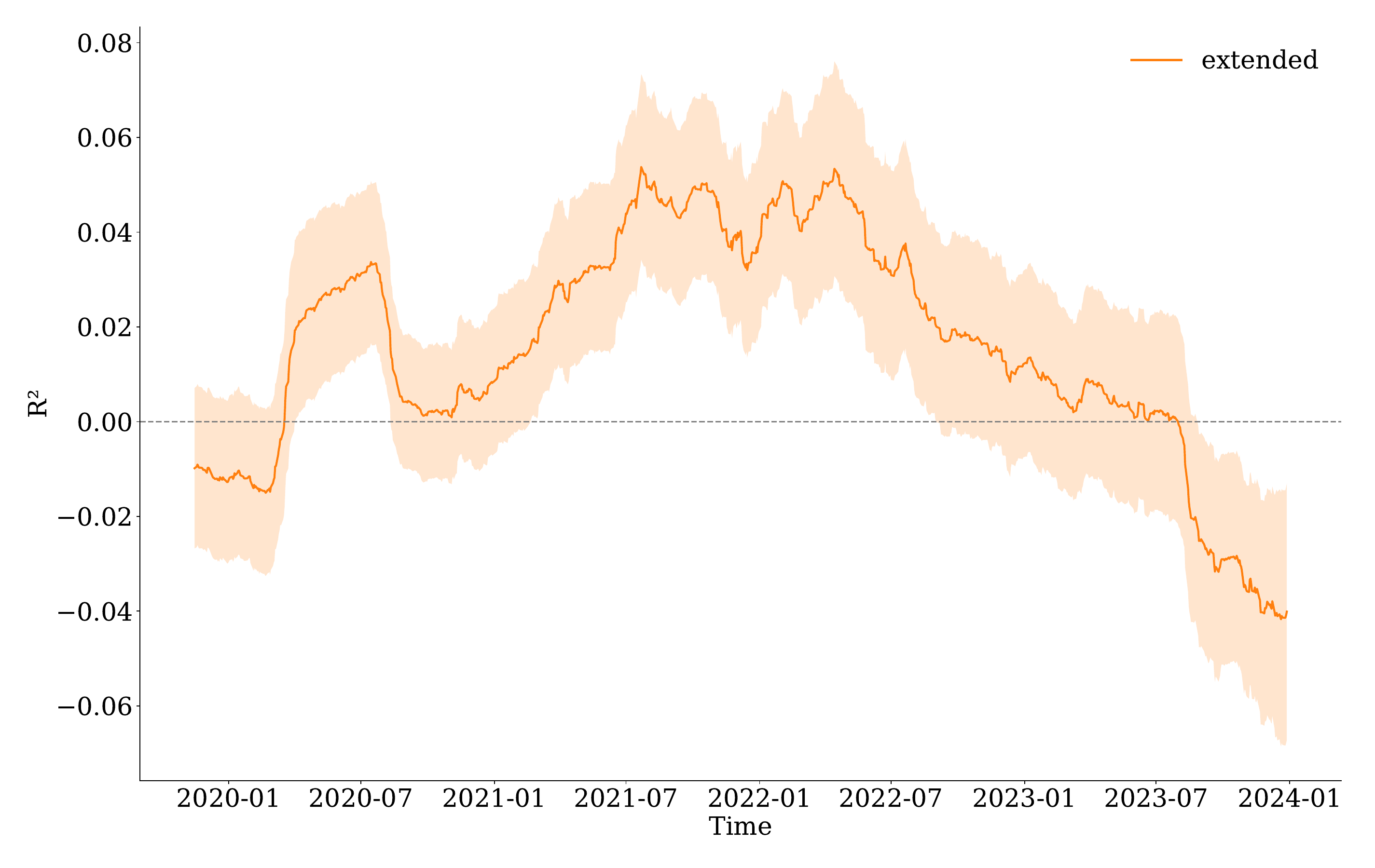}
    \caption{Evidence against the null `no added factors' ($R^2$) for the extended risk model. At each of $30$ replications and at each time, assets are randomly split into train and test sets (90/10), and the out-of-sample $R^2$ \eqref{eq:r2_res_exttt} is computed for each day. Lines report the rolling mean (window = 100 days) of the cross-replication average $R^2$. The shaded bands indicate the $\pm 1$ rolling mean of the standard deviation of the average cross-replication $R^2$ (window = 100 days), providing a smoothed measure of sampling variability across random splits.}
    \label{fig:mosaic}
\end{figure}

% \begin{table}[ht]
% \centering
% \begin{tabular}{lcc}
% \toprule
% Model
% & \makecell{Point estimate} &
% \makecell{Dispersion} \\
% \midrule
% Extended & \textbf{0.0125} & 0.002 \\
% Extended (permuted) & $-$0.0168 & 0.0005 \\
% \bottomrule
% \end{tabular}
% \caption{Evidence ($R^2$) against the null `no added factors'. The reported point estimate corresponds to the time average of the cross-replication mean. The reported dispersion measure is the standard deviation of the average across time of the cross-replication mean. A positive point estimate indicates the existence of added factors in the returns. }
% \label{tab:mosaic}
% \end{table}

\subsection{Log Likelihood and Regret}

% \ejc{I do not know that this should be the first section. If returns are not Gaussian, people may not be impressed with this. (fixed)}

% MK: should we be dropping additive and multiplicative constants here as well? should we use \ell? I opted against dropping them here and in eq(7). (Fixed)
\rf{Every} day $t$, we compute the normalized log-likelihood \rf{(i.e., log-likelihood divided by the number of assets)} of the next return, $x_{t+1}$, under our current model\rev{\begin{equation*}
    \ell(\hat{\Sigma}_t, x_{t+1}) = -\dfrac{1}{2n} \left( n \log (2\pi) + \log \det \hat{\Sigma}_t + x_{t+1}^T \hat{\Sigma}_t^{-1} x_{t+1}\right).
\end{equation*}}

% This metric is based on the assumption that the daily asset returns are Gaussian, which, as we saw earlier, is a reasonable assumption. Furthermore, our method for extending a base risk model is derived under Gaussian assumptions. A large value of log likelihood indicates a good statistical fit and suggests that the data is highly likely under our model.

Given the series of returns, the best constant predictor (in terms of log-likelihood) is the empirical sample covariance
\begin{equation*}
    \Sigma^\mathrm{emp} = \dfrac{1}{T}\sum_{\tau =1}^T x_\tau x_\tau^\top,
\end{equation*}
with normalized log-likelihood $\ell_t(\Sigma^\mathrm{emp}, x_{t+1})$.
Here $\tau=1$ corresponds to 2019-06-27 and $\tau=T$ corresponds to 2023-12-28. We therefore also report the difference
\begin{equation*}
    \Delta_t = \ell(\Sigma^\mathrm{emp}, x_{t+1}) - \ell_t(\hat{\Sigma}_t, x_{t+1}).
\end{equation*}
This is a measure of regret and measures how much the covariance estimate underperforms the best possible constant covariance predictor. We want our covariance estimate to have small regret. Note that the regret can be negative, because our covariance is time-dependent, i.e., a time-dependent covariance can outperform $\Sigma^\mathrm{emp}$.

Our results are included in \Cref{tab:loglike}. We show both the average normalized log-likelihood and regret over the evaluation period, but also the standard deviation of these averages. The extended model obtains a lift in both log-likelihood and regret.
% The extended (permuted) model is worse than the base model.
% Based on the reported standard deviations, the lifts are considered statistically significant.

% \begin{table}[ht]
% \centering
% \begin{tabular}{lcccc}
% \toprule
% Model
% & \makecell{Average \\ Log-Likelihood}
% & \makecell{Std. of Average \\ Log-Likelihood}
% & \makecell{Average \\ Regret}
% & \makecell{Std. of Average \\ Regret} \\
% \midrule
% Base     & 2330.53  & 33.77  & 491.38  & 33.15 \\
% Extended & \textbf{2371.73} &  27.99 &  \textbf{450.18} &  27.30 \\
% Extended (permuted) & 2273.73 &  38.59 & 548.18 & 38.01 \\
% \bottomrule
% \end{tabular}
% \caption{Normalized log likelihood and normalized log likelihood regret metrics. The normalization is by the number of assets.}
% \label{tab:loglike}
% \end{table}

\begin{table}[ht]
\centering
\begin{tabular}{lcccc}
\toprule
Model
& \makecell{Average \\ Log-Likelihood}
& \makecell{Std. of Average \\ Log-Likelihood}
& \makecell{Average \\ Regret}
& \makecell{Std. of Average \\ Regret} \\
\midrule
Base     & 2.679  & 0.039  & 0.565  & 0.038 \\
Extended & \textbf{2.726} & 0.032 & \textbf{0.517} & 0.031 \\
% Extended (permuted) & 2.613 & 0.0444 & 0.630 & 0.0437 \\
\bottomrule
\end{tabular}
\caption{Normalized log-likelihood and normalized log-likelihood regret metrics. The normalization is by the number of assets.}
\label{tab:loglike}
\end{table}

\subsection{Whitened Returns}\label{sec:whitened}

For each risk model, we compute out-of-sample whitened returns and their empirical correlation matrix.
We judge the risk model by the discrepancy between the empirical correlation matrix and the implied (by the risk model) correlation matrix of the whitened returns.

% We consider two null hypotheses. Each hypothesis corresponds to a risk model--\textit{base} or \textit{extended}--being true.

In detail, consider an arbitrary risk model $\hat{\Sigma}_t$ for every time $t$.
% This would correspond to either the \textit{base}, \textit{extended}, or \textit{extended (permuted)} model.
Under the hypothesis that the risk model is true and the returns are zero-mean, then $\mathbf{cov}(x_t) = \hat{\Sigma}_t$ and $\mathbf{cov}(\hat{\Sigma}_t^{-1/2} x_t)=I$, where $I$ is the identity.

We now collect the out-of-sample whitened returns $\lbrace \hat{\Sigma}^{-1/2}_t x_{t+1} \rbrace_{t=1}^T$ that should have the identity as the covariance under this hypothesis. We compute the Frobenius norm between the empirical correlation matrix ($n \times n$ matrix) of the whitened returns $\lbrace \hat{\Sigma}^{-1/2}_t x_{t+1} \rbrace_{t=1}^T$ and their theoretical correlation matrix assuming the risk model is true, i.e., the identity matrix $I$. We report the Frobenius norm for each risk model in \Cref{tab:whitened}, but normalized by \rf{$\sqrt{n (n-1)}$}.
A larger normalized Frobenius norm indicates a greater discrepancy between the risk model and the covariance structure of the realized returns. The extended model achieves the best fit.

\begin{table}[ht]
\centering
\begin{tabular}{lcc}
\toprule
Model
& \makecell{(normalized) Frobenius norm of\\ correlation matrix to identity} \\
\midrule
Base     &  0.077    \\
Extended &  \textbf{0.056} \\
% Extended (permuted) &  0.078 \\
\bottomrule
\end{tabular}
\caption{Results for the whitened returns. The normalization is by $\sqrt{n (n-1)}$, where $n$ is the number of assets.}
\label{tab:whitened}
\end{table}

% \ejc{I don't understand why we discuss normal distributions in this section. The entire discussion is about second moments. (fixed)}

% This implies that the correlation matrix is also $I$. We test how close the correlation matrix of whitened returns, $\mathbf{corr}\left(\lbrace \hat{\Sigma}^{-1/2}_t x_{t+1} \rbrace_t \right)$, is to the identity matrix. We report the Frobenius norm
% \begin{equation}
%     \lVert I_n - \mathbf{corr}\left(\lbrace \hat{\Sigma}^{-1/2}_t x_{t+1} \rbrace_t \right) \rVert_F,
% \end{equation}
% which measures how far the off-diagonal elements of the correlation matrix are from zero.
% % , and the maximum such deviation
% % \begin{equation}
%     % \max_{i \neq j}\ \lvert \mathbf{corr}\left(\lbrace \hat{\Sigma}^{-1/2}_t x_{t+1} \rbrace_t \right)_{i, j} \rvert.
% % \end{equation}
% Theoretically, the Frobenius norm should be zero. Our results are shown in \Cref{tab:whitened}. The extended model results in whitened out-of-sample returns that are closer to being uncorrelated.

% \begin{table}[ht]
% \centering
% \begin{tabular}{lcc}
% \toprule
% Model
% & \makecell{Frobenius norm of correlation\\ matrix to identity} \\
% \midrule
% Base     & 67.15    \\
% Extended & \textbf{49.20} \\
% Extended (permuted) & 68.36 \\
% \bottomrule
% \end{tabular}
% \caption{Results for the whitened returns.}
% \label{tab:whitened}
% \end{table}

\subsection{Whitened Residuals}
We perform a similar procedure to the previous section, but instead of looking at out-of-sample whitened returns, we focus on out-of-sample whitened residuals.
We describe our approach using an arbitrary risk model $\hat{\Sigma}_t$ for every time $t$.
% This would correspond to either the \textit{base}, \textit{extended}, or \textit{extended (permuted)} model.

At time $t$, we solve \rev{the out-of-sample ordinary least squares problem}
\rev{\begin{equation}\label{eq:ltsq}
    \minimize_{s_t} \lVert x_{t+1} - \hat{F}_t s_t \rVert_2^2
    % \lVert r_{t+1} - \hat{F}_t s_t \rVert_2^2
\end{equation}}
and compute the residual
\begin{equation*}
    \hat{\delta}_t = x_{t+1} - \hat{F}_t {\hat{s}_t},
\end{equation*}
where $\hat{s}_t$ is the solution of \eqref{eq:ltsq}.

We assume that the risk model $\hat{\Sigma}_t$ holds for the returns, i.e., the returns follow
% \ejc{Again, I don't understand the role of the normal distribution. (fixed)}
\[x_t = \hat{F}_t s_t + \epsilon_t, \quad \mathbf{cov}(s_t) = I, \quad \mathbf{cov}(\epsilon_t)=\hat{D}_t.
\]
By letting $M_t = I - \hat{F}_t (\hat{F}_t^\top \hat{F}_t)^{-1} \hat{F}_t^\top$, it also holds that the in-sample ordinary least squares residuals $e_t$ (where we solve \eqref{eq:ltsq} but replacing $x_{t+1}$ with $x_t$) are equal to $M_t \epsilon_t$. Since $M_t$ is an orthogonal projection matrix it can be written as $M_t = U_t U_t^\top$, where $U_t$ is orthonormal. Let $z_t = U_t^\top e_t$. Then $\mathbf{cov}(z_t) = U_t^\top \hat{D}_t U_t$ and
\begin{equation*}
    \mathbf{cov}\left( (U_t^\top \hat{D}_t U_t)^{-1/2} z_t \right) = I.
\end{equation*}
We collect the out-of-sample whitened residuals
\begin{equation*}
    \lbrace (U_t^\top \hat{D}_t U_t)^{-1/2} U_t^\top \hat{\delta}_t \rbrace_{t=1}^T
\end{equation*}
and compute the Frobenius norm between the empirical correlation matrix of the out-of-sample whitened residuals and their theoretical correlation matrix under the assumption that the risk model holds, i.e., the identity $I$.
For each of the risk models, we report the normalized \rf{(by $\sqrt{n (n-1)}$)} Frobenius norm in \Cref{tab:residuals}.
A larger normalized Frobenius norm indicates a greater discrepancy between the risk model and the covariance structure of the realized returns. The extended model achieves the best statistical fit.

\begin{table}[ht]
\centering
\begin{tabular}{lcc}
\toprule
Model
& \makecell{(normalized) Frobenius norm of\\ correlation matrix to identity} \\
\midrule
Base     & 0.072  \\
Extended  & \textbf{0.050} \\
% Extended (permuted) & 0.057 \\
\bottomrule
\end{tabular}
\caption{Results for the whitened residuals. The normalization is by $\sqrt{n (n-1)}$, where $n$ is the number of assets.}
\label{tab:residuals}
\end{table}

\section{Discussion}\label{sec:conclusion}
We proposed a methodology to refine an available low-rank-plus-diagonal covariance model and extend it with additional statistical factors. Our method is based on maximum likelihood estimation and is an instance of the expectation-maximization algorithm.  This approach is particularly useful if the provided covariance
 model is only infrequently updated, but we would like to incorporate information in shorter time horizons as well. We empirically show that this modification can improve the model's statistical fit \rev{even when the provided covariance model is of high quality}. Our approach is also able to handle missing returns data. 
 However, we note that adding factors does not necessarily improve statistical fit \cite{spector2024mosaic}, as shown in 
 \cref{appendix:additional-analyses}. \rev{Furthermore, the effectiveness of our approach depends on the quality of the factors in the original covariance model, namely the quality of the provided factor exposure matrix.}

Although we cover additional analyses in \cref{appendix:additional-analyses}, there are various
interesting directions for future work.
We leave the problem of selecting the number of additional factors as future work, although we do note that a value in the range 2-10 typically worked well for a base risk model of roughly 70 factors.
One simple approach for determining the number of additional factors is to evaluate out-of-sample $R^2$ and select the number of additional factors corresponding to the cutoff beyond which $R^2$ decreases.
An important aspect not covered in this work is interpreting the added statistical factors.
A first attempt is carried out in \cref{appendix:additional-analyses}. This can also be done by looking at the cross-sectional correlation of the added factors with existing themes \cite{candes2025thematic} or by querying a large language model.

We plan to combine our
extension methodology with factor selection in the base model. This can be helpful if the base model has many factors, but only a few of them are relevant. We can reduce the number of base factors using,
e.g., forward selection, and then apply our extension methodology to the reduced base model. A related approach is carried out in 
\cref{appendix:additional-analyses}.
We can also consider learning an $F_2$ with a specified range, if we wish to capture specific factors.

% \rf{Deriving the EM update for a block-diagonal matrix $D$ would be useful for extending a factor model that captures both equities and exchange traded funds (ETFs).}

% \rf{Future work will also deal with interpreting the added statistical factors. This can be done by looking at the cross-sectional correlation of the added factors with existing themes or by querying a large language model.}

\rf{Finally, we plan to construct portfolios via Markowitz optimization \cite{boyd2024markowitz} where we either use the extended or the base risk model for the asset return covariance.
Preliminary results suggest that the extended model produces a Pareto frontier that dominates that of the base model in the space of realized returns and target volatilities.}
Menchero and Lazanas~\cite{menchero2024evaluating} advocate comparing risk
models through the out-of-sample volatility (the lower the better) of minimum-variance portfolios
with unit exposure to a common alpha signal, since minimum-variance optimization with the true covariance
minimizes realized portfolio variance under this constraint, even when the signal
is misspecified.
Comparisons based on Sharpe need not isolate
risk-model quality, because misspecified alpha can produce different
expected returns across these portfolios (different numerators in the Sharpe ratios), while estimated Sharpe 
ratios are
also subject to substantial sampling noise.

\appendix  %This command ends the counting of sections.
\section{Different Representations of the Extended Risk Model}
\label{appendix:param}
We prove the equivalence between two different representations of the family of \textit{extended} risk models.
\begin{lemma}
    Consider the family of risk models \eqref{eq:family1}. This is equivalent to the family of risk models given by
    \begin{equation}
    \label{eq:family3}
    \hat{\Sigma}\left(F_{1, \perp}, \Sigma_f, D\right) =  \begin{bmatrix}
        F_1 & F_{1, \perp}
    \end{bmatrix}\Sigma_f\begin{bmatrix}
        F_1^\top\\
        F_{1, \perp}^\top
    \end{bmatrix} + D,
\end{equation}
where $F_{1, \perp}$ is a matrix whose columns are a basis for the orthogonal complement of the range of $F_1$ (i.e., $F_{1, \perp}^\top F_1 = 0$).
\end{lemma}
\begin{proof}
    Consider a risk model $\hat{\Sigma}\left(F_{1, \perp}, \Sigma_f, D\right)$  in the family \eqref{eq:family3}. It is obviously included in family \eqref{eq:family1}.

    Now consider a risk model $\hat{\Sigma}\left(F_{2}, \Sigma_f, D\right)$ in family \eqref{eq:family1}. We can write
    \begin{equation*}
        F_2 = F_1 U + F_{1, \perp} V,
    \end{equation*}
    where the columns of $F_{1, \perp}$ are a basis of the orthogonal complement of the range of $F_1$. We can then rewrite the risk model as
    \begin{equation*}
    \hat{\Sigma}\left(F_{2}, \Sigma_f, D\right) =
    \begin{bmatrix}
        F_1 & F_{1, \perp}
    \end{bmatrix} \begin{bmatrix}
        I & U\\
        0 & V
    \end{bmatrix}\Sigma_f
    \begin{bmatrix}
        I & 0\\
        U^\top & V^\top
    \end{bmatrix}
    \begin{bmatrix}
        F_1 & F_{1, \perp}
    \end{bmatrix}^\top +D.
    \end{equation*}
    Different values of $F_2$ correspond to different $U, V$. Therefore, the two families are the same.
\end{proof}
\section{Equivalence of Objectives}\label{appendix:equivalence}
Let $p_\theta(x)$ be the probability density function of the Gaussian $\mathcal{N}(0, \tF \tOmega \tF^\top +D)$:
\begin{equation*}
\begin{aligned}
    \log p_\theta(x) &= - \dfrac{1}{2} x^\top \left( \tF \tOmega \tF^\top +D \right)^{-1} x - \dfrac{1}{2} \log \det \left( \tF \tOmega \tF^\top +D \right) - \dfrac{1}{2} n \log (2 \pi)\\
    &= - \dfrac{1}{2} \tr \left( \left( \tF \tOmega \tF^\top +D \right)^{-1} xx^\top\right) - \dfrac{1}{2} \log \det \left( \tF \tOmega \tF^\top +D \right) - \dfrac{1}{2} n \log (2 \pi).
\end{aligned}
\end{equation*}
We show that the following two problems:
\begin{equation*}
    \maximize_\theta \sum_{\tau \leq T} w_\tau \log p_{\theta}(x_{\tau}),
\end{equation*}
\begin{equation}
\label{eq:main_KL}
    \minimize_\theta \mathrm{KL}\left(\mathcal{N}(0, C) \quad \Vert \quad \mathcal{N}(0, \tF \tOmega \tF^\top +D)\right)
\end{equation}
are equivalent, i.e., they have the same set of solutions $\theta^\star$, for
\begin{equation*}
    C = \sum_{\tau \leq T} w_\tau x_\tau x_\tau^\top.
\end{equation*}
By the linearity of the trace operator and the fact that $\sum_{\tau \leq T} w_\tau=1$, we get
\begin{equation*}
    \begin{aligned}
    \sum_{\tau \leq T} & w_\tau \log p_{\theta}(x_\tau)  \\
    &=- \dfrac{1}{2} \tr \left( \left( \tF \tOmega \tF^\top +D \right)^{-1} \sum_{\tau \leq T} w_\tau x_\tau x_\tau^\top\right) - \dfrac{1}{2} \log \det \left( \tF \tOmega \tF^\top +D \right) - \dfrac{1}{2} n \log (2 \pi) \\
    &=- \dfrac{1}{2} \tr \left( \left( \tF \tOmega \tF^\top +D \right)^{-1} C \right) - \dfrac{1}{2} \log \det \left( \tF \tOmega \tF^\top +D \right) - \dfrac{1}{2} n \log (2 \pi).
    \end{aligned}
\end{equation*}
Problem \eqref{eq:main_KL} is equivalent (has the same solution) to
\begin{equation}
    \label{eq:_main}
    \maximize_\theta \mathbb{E}_{x \sim \mathcal{N}(0, C)} \log p_\theta(x),
\end{equation}
by expanding the KL divergence objective.
% MK: maybe just define \ell_theta(x) and mention the dropping of additive and multiplicative constants (Fixed, no dropping)
By the cyclic property and the linearity of the trace operator,
\begin{equation*}
\begin{aligned}
    \mathbb{E}_{x \sim \mathcal{N}(0, C)} \log p_\theta(x) = &- \dfrac{1}{2} \tr\left( \left( \tF \tOmega \tF^\top +D \right)^{-1} C \right) \\
    &- \dfrac{1}{2} \log \det \left( \tF \tOmega \tF^\top +D \right) - \dfrac{1}{2} n \log (2 \pi).
\end{aligned}
\end{equation*}
This proves the claim and implies that our method can be thought of as maximizing log-likelihood or, equivalently, as minimizing the KL divergence with a time-weighted empirical covariance matrix.

As an aside, we should note that the minimizer of the Gaussian KL objective is equivariant under congruence transformations. Therefore, although not a focus of this paper, we could equivalently work with standardized returns.

We further observe that for any zero-mean distribution $\mathbb{P}$ with covariance $C$ it holds that
\begin{equation*}
    \mathbb{E}_{x \sim \mathcal{N}(0, C)} \log p_\theta(x) = \mathbb{E}_{x \sim \mathbb{P}} \log p_\theta(x).
\end{equation*}
This justifies talking about problem \eqref{eq:_main} as maximum likelihood estimation, as we are maximizing the expected log-likelihood of the Gaussian $\mathcal{N}(0,  \tF \tOmega \tF^\top +D) $ under any  distribution $\mathbb{P}$.

\section{The EM Update for Fully Observed Returns}
\label{appendix:em}
Our derivation of the iterative EM algorithm focuses on optimizing a series of lower bounds for problem
\begin{equation*}
    \maximize_\theta \sum_{\tau \leq T} w_\tau \log p_{\theta}(x_{\tau}),
\end{equation*}
which is equivalent to
\begin{equation*}
    \maximize_\theta \mathbb{E}_{x \sim \mathcal{N}(0, C)} \log p_\theta(x),
\end{equation*}
as shown in \Cref{appendix:equivalence}. In \Cref{appendix:equivalence}, we also showed that
these problems are equivalent to
\begin{equation*}
    \maximize_\theta \mathbb{E}_{x \sim \mathcal{P}} \log p_\theta(x),
\end{equation*}
where $\mathcal{P}$ is the empirical data distribution (with covariance $C$ \eqref{eq:C}).
% We showed the equivalence of these problems with \eqref{eq:main_appp} (which is the weighted Gaussian log-likelihood estimation problem for fully observed returns) in \Cref{appendix:equivalence}.

\subsection{A Lower Bound of the Objective}
We fix $\theta$, i.e., $F_2, \Omega$, and $D$. Let
\begin{equation*}
    S =\begin{bmatrix}
        S_1 \\
        S_2
    \end{bmatrix}\sim \mathcal{N}(0, \begin{bmatrix}
    \Omega & 0\\
    0 & I
\end{bmatrix}), \quad E \sim \mathcal{N}(0, D),
\end{equation*}
where $S_1 \in \mathbb{R}^{n_1}, S_2 \in \mathbb{R}^{n_2}$ and $E \in \mathbb{R}^n$ are random variables. Further, we suppose $S$ and $E$ are independent, i.e., $S \perp E$. We define
\begin{equation*}
    X = \begin{bmatrix}
        F_1 & F_2
    \end{bmatrix}\begin{bmatrix}
        S_1 \\
        S_2
    \end{bmatrix} + E.
\end{equation*}
It follows that
\begin{equation*}
    \begin{bmatrix}
        S \\
        X
    \end{bmatrix} \sim \mathcal{N}\left(0, \begin{bmatrix}
        \tOmega & \tOmega \tF^\top \\
        \tF \tOmega & \tF \tOmega \tF^\top +D
    \end{bmatrix} \right),
\end{equation*}
where
\[
\Tilde{F} = \begin{bmatrix}
    F_1 & F_2
\end{bmatrix}, \quad \Tilde{\Omega} =\begin{bmatrix}
    \Omega & 0 \\
    0 & I
\end{bmatrix}.
\]
Marginally $X$ is distributed according to $\mathcal{N}(0, \tF \tOmega \tF^\top +D)$ with density
% MK: Not sure what is going on here for the rest of the paragraph. This is the definition of the marginal density, no? Then it says "In other words, p_theta(x) is the marginal density of X" which seems redundant. (Fixed)
\begin{equation*}
    p_\theta(x) = \int_s p_\theta(x,s) ds,
\end{equation*}
where $p_\theta(x,s)$ is the joint density of $(X,S)$ evaluated at $X=x$ and $S=s$.
% In other words, $p_\theta(x)$ is the marginal density of $X$.

We now also fix $x$. Let $q(s; x)$ be a probability density function in $s$, parametrized by $x$. Then
\begin{equation}
\label{eq:jensen}
\begin{aligned}
    \log p_\theta(x) &= \log \int_s  p_\theta(x,s) ds \\
    &= \log \int_s q(s; x) \dfrac{ p_\theta(x,s)}{q(s; x)} ds = \log \mathbb{E}_{s \sim q(s;x)} \dfrac{p_\theta(x,s)}{q(s;x)}
    \geq \mathbb{E}_{s \sim q(s;x)} \log \dfrac{p_\theta(x,s)}{q(s;x)},
\end{aligned}
\end{equation}
where the inequality follows from Jensen's inequality for the concave function $\log$. For brevity, we may write $q(s;x)$ instead of $s \sim q(s;x)$. Because eq.~\eqref{eq:jensen} holds for all $x$, it is true that
\begin{align}
\label{eq:ineq}
    \mathbb{E}_{x \sim \mathcal{N}(0, C)} \log p_\theta(x) \geq \mathbb{E}_{x \sim \mathcal{N}(0, C)} \mathbb{E}_{q(s; x)} \log \dfrac{p_\theta(x,s)}{q(s;x)}.
\end{align}
This implies that
\begin{equation*}
    \maximize_\theta \mathbb{E}_{x \sim \mathcal{N}(0, C)} \log p_\theta(x) \geq \maximize_\theta \mathbb{E}_{x \sim \mathcal{N}(0, C)} \mathbb{E}_{q(s;x)} \log \dfrac{p_\theta(x,s)}{q(s;x)},
\end{equation*}
where the optimal point of the right-hand side is the same as that of
\begin{equation}\label{eq:RHS}
    \maximize_\theta \mathbb{E}_{x \sim \mathcal{N}(0, C)} \mathbb{E}_{q(s;x)} \log p_\theta(x,s).
\end{equation}
Notice that we can write
\begin{equation*}
    p_\theta(x,s) =p_\theta(X=x, S=s) = p_\theta(E = x-\tF s, S=s) = p_\theta(E = x-\tF s) p_\theta(S=s),
\end{equation*}
because $S \perp E$.
Therefore, up to an additive constant,
\begin{align}
    \label{eq:p_xs}
    2\log p_\theta(x,s) = -(x-\tF s)^\top D^{-1} (x-\tF s) - \log \det D - s^\top \tOmega^{-1} s - \log \det \tOmega.
\end{align}
By expanding eq.~\eqref{eq:p_xs} and using
\begin{equation*}
    \tOmega^{-1} = \begin{bmatrix}
        \Omega^{-1} & 0\\
        0 & I
    \end{bmatrix},
\end{equation*}
we get
\begin{equation*}
\begin{aligned}
    &- 2\log p_\theta(x,s) = \\
    &x^\top D^{-1}x - 2\begin{bmatrix}
        s_1^\top & s_2^\top
    \end{bmatrix}\begin{bmatrix}
        F_1^\top \\
        F_2^\top
    \end{bmatrix}D^{-1} x  + s_1^\top F_1^\top D^{-1} F_1 s_1 + 2 s_2^\top F_2^\top D^{-1} F_1 s_1  + s_2^\top F_2^\top D^{-1} F_2 s_2 \\
    &+s_1^\top \Omega^{-1} s_1 +s_2^\top s_2 + \log \det D + \log \det \Omega.
\end{aligned}
\end{equation*}
For any density $q(s;x)$, ignoring all terms that do not depend on $F_2, \Omega$ or $D$, using the fact that the trace of a scalar equals the scalar, and the cyclic property of the trace, we get
\begin{equation}
\label{eq:obj}
\begin{aligned}
    \mathbb{E}_{q(s;x)} &\left[ - 2\log p_\theta(x,s)\right] = \\
    &\tr(D^{-1}xx^\top) + \tr(D^{-1} F_1 \mathbb{E}_{q(s;x)} \left[ s_1 s_1^\top \right] F_1^\top) - 2\tr(D^{-1} \mathbb{E}_{q(s;x)}\left[ x s_1^\top\right]F_1^\top)\\
    &-2\tr(D^{-1} \mathbb{E}_{q(s;x)} \left[ x s_2^\top  \right] F_2^\top) + 2 \tr(D^{-1}F_1 \mathbb{E}_{q(s;x)}\left[ s_1 s_2^\top \right] F_2^\top) \\
    &+ \tr(D^{-1}F_2 \mathbb{E}_{q(s;x)}\left[ s_2 s_2^\top\right] F_2^\top)\\
    &+\tr(\Omega^{-1} \mathbb{E}_{q(s;x)}\left[ s_1 s_1^\top \right])\\
    &+ \log \det D + \log \det \Omega.
\end{aligned}
\end{equation}

At this point, we have an explicit expression for the objective of problem \eqref{eq:RHS}.

\subsection{Deriving the EM Update}

The EM algorithm is an iterative algorithm that solves a sequence of problems of the form \eqref{eq:RHS}. At each iteration the density $q(s;x)$ is updated. Assuming that at iteration $k$ we are provided with $\theta_k$, the EM algorithm computes $\theta_{k+1}$ as
\begin{equation}
\label{eq:emem}
\begin{aligned}
    \theta_{k+1} &= \argmax_\theta \mathbb{E}_{x \sim \mathcal{N}(0, C)} \mathbb{E}_{s \sim q_k(s;x)} \log p_\theta(x,s)\\
    &=\argmin_\theta \mathbb{E}_{x \sim \mathcal{N}(0, C)} \mathbb{E}_{s \sim q_k(s;x)}  \left[ - 2\log p_\theta(x,s)\right],
\end{aligned}
\end{equation}
where $q_k(s; x) = p_{\theta_k}(S=s \mid X=x)$.
We therefore need to compute
\begin{itemize}
    \item $\mathbb{E}_{q_k(s;x)} \left[s_1 \right]=: \mathbb{E}\left[ S_1 \mid X=x; \theta_k\right]$
    \item $\mathbb{E}_{q_k(s;x)} \left[s_2 \right]=: \mathbb{E}\left[ S_2 \mid X=x; \theta_k\right]$ ,
    \item $\mathbb{E}_{q_k(s;x)} \left[ s_1 s_2^\top\right]=: \mathbb{E} \left[ S_1 S_2^\top \mid X=x; \theta_k\right]$,
    \item $\mathbb{E}_{q_k(s;x)} \left[ s_2 s_2^\top \right] =: \mathbb{E} \left[ S_2 S_2^\top \mid X=x; \theta_k\right]$,
    \item $\mathbb{E}_{q_k(s;x)} \left[ s_1 s_1^\top \right] =: \mathbb{E} \left[ S_1 S_1^\top \mid X=x; \theta_k\right]$.
\end{itemize}
This is easy to do under our Gaussian assumptions. Note that
\begin{equation*}
    \begin{bmatrix}
        S_1 \\
        S_2
    \end{bmatrix} \mid X=x; \theta_k \sim \mathcal{N}(L_k x, G_k).
\end{equation*}
We can compute $L_k$ and $G_k$ using the density $p_{\theta_k}(S=s \mid X=x)$. Because $S \mid X=x$ is Gaussian:
\begin{itemize}
    \item $\mathbb{E}\left[S \mid X=x; \theta_k \right] = \argmax_s \log p_{\theta_k}(S=s \mid X=x) = \argmax_s \log p_{\theta_k}(x,s)$. By taking the derivative of \eqref{eq:p_xs} with respect to $s$ and setting it to zero, we get the necessary optimality condition
    \begin{equation*}
        - \tF_k^\top D_k^{-1} x + \tF_k^\top D_k^{-1} \tF_k s + \tOmega_k^{-1} s = 0 \Leftrightarrow s = \left( \tF_k^\top D_k^{-1} \tF_k + \tOmega_k^{-1} \right)^{-1} \tF_k^\top D_k^{-1} x.
    \end{equation*}
    This implies that
    \begin{equation*}
        \mathbb{E}\left[S \mid X=x; \theta_k \right] = \left( \tF_k^\top D_k^{-1} \tF_k + \tOmega_k^{-1} \right)^{-1} \tF_k^\top D_k^{-1} x.
    \end{equation*}
    Quantities indexed by the iteration $k$ are evaluated using the parameter value $\theta_k$, e.g., $\Tilde{F}_k$ is the value of $\Tilde{F}$ at iteration $k$.
    \item The inverse of the conditional covariance is given by the Hessian of $-\log p_{\theta_k}(S=s\mid X=x)$ with respect to $s$ (that is equal to the Hessian of $-\log p_{\theta_k}(S=s, X=x)$ with respect to $s$), i.e.,
    \begin{equation*}
        G_k^{-1} = \tF_k^\top D_k^{-1} \tF_k + \tOmega_k^{-1}.
    \end{equation*}
\end{itemize}
We combine the two results and get that
\begin{equation*}
    L_k = G_k \tF_k^\top D_k^{-1}, \quad G_k^{-1} = \tF_k^\top D_k^{-1} \tF_k + \tOmega_k^{-1}.
\end{equation*}
Furthermore, by the definition of the conditional covariance, it holds that
\begin{equation*}
    \mathbb{E} \left[ SS^\top \mid X=x; \theta_k\right] = G_k + \mathbb{E}\left[ S \mid X=x; \theta_k\right]\mathbb{E}\left[ S \mid X=x; \theta_k\right]^\top,
\end{equation*}
which can equivalently be written as
\begin{align*}
    \begin{bmatrix}
        \mathbb{E} \left[ S_1S_1^\top \mid X=x; \theta_k\right] & \mathbb{E} \left[ S_1S_2^\top \mid X=x; \theta_k\right]\\[0.8em]
        \mathbb{E} \left[ S_2S_1^\top \mid X=x; \theta_k\right] & \mathbb{E} \left[ S_2S_2^\top \mid X=x; \theta_k\right]
    \end{bmatrix} = G_k + L_k xx^\top  L_k^\top.
\end{align*}
Notice that we can select the appropriate block from $\mathbb{E} \left[ SS^\top \mid X=x; \theta_k\right]$ by multiplying left and right with the appropriate (block-selecting) matrices. For example:
\begin{itemize}
    \item
\begin{equation*}
        \mathbb{E} \left[ S_1S_1^\top \mid X=x; \theta_k\right] = B_{l, 11} (G_k + L_k xx^\top  L_k^\top) B_{r, 11},
    \end{equation*}
    where
    \begin{equation*}
        B_{l, 11} = \begin{bmatrix}
            I_{n_1} & 0_{n_1 \times n_2}
        \end{bmatrix}, \quad
            B_{r, 11} = \begin{bmatrix}
                I_{n_1} \\
                0_{n_2 \times n_1}
            \end{bmatrix},
    \end{equation*}
    and $I_{n_1}$ is the identity matrix of dimension $n_1$.
\item
\begin{equation*}
        \mathbb{E} \left[ S_2S_2^\top \mid X=x; \theta_k\right] = B_{l, 22} (G_k + L_k xx^\top  L_k^\top) B_{r, 22},
    \end{equation*}
    where
    \begin{equation*}
        B_{l, 22} = \begin{bmatrix}
            0_{n_2 \times n_1} & I_{n_2}
        \end{bmatrix}, \quad
            B_{r, 22} = \begin{bmatrix}
                0_{n_1 \times n_2} \\
                I_{n_2}
            \end{bmatrix}.
    \end{equation*}
\item
\begin{equation*}
        \mathbb{E} \left[ S_1S_2^\top \mid X=x; \theta_k\right] = B_{l, 12} (G_k + L_k xx^\top  L_k^\top) B_{r, 12},
    \end{equation*}
    where
    \begin{equation*}
        B_{l, 12} = \begin{bmatrix}
            I_{n_1} & 0_{n_1 \times n_2}
        \end{bmatrix}, \quad
            B_{r, 12} = \begin{bmatrix}
                0_{n_1 \times n_2} \\
                I_{n_2}
            \end{bmatrix}.
    \end{equation*}
\end{itemize}
Similarly,
\begin{itemize}
\item \begin{equation*}
    \mathbb{E}\left[ S_1 \mid X=x; \theta_k\right] = B_{1} L_k x,
\end{equation*}
where $B_1 = \begin{bmatrix}
    I_{n_1} & 0_{n_1 \times n_2}
\end{bmatrix}$.
    \item \begin{equation*}
    \mathbb{E}\left[ S_2 \mid X=x; \theta_k\right] = B_{2} L_k x,
\end{equation*}
where $B_2 = \begin{bmatrix}
    0_{n_2 \times n_1} & I_{n_2}
\end{bmatrix}$.
\end{itemize}
By combining the above results with \eqref{eq:obj}, we get that
\begin{equation*}
\begin{aligned}
    \mathbb{E}_{s \sim q_k(s;x)}  &\left[ - 2\log p_\theta(x,s)\right] = \\
    &\tr(D^{-1}xx^\top) + \tr(D^{-1} F_1 B_{l, 11} (G_k + L_k xx^\top  L_k^\top) B_{r, 11} F_1^\top) \\
    &- 2\tr(D^{-1}xx^\top L_k^\top B_1^\top F_1^\top)\\
    &-2\tr(D^{-1} x x^\top L_k^\top B_2^\top F_2^\top)\\
    &+ 2 \tr(D^{-1}F_1 B_{l, 12} \left( G_k  + L_k x x^\top L_k^\top \right)B_{r, 12}F_2^\top) \\
    &+ \tr(D^{-1}F_2 B_{l, 22}\left( G_k  + L_k x x^\top L_k^\top \right)B_{r, 22} F_2^\top)\\
    &+\tr(\Omega^{-1}  B_{l, 11} \left( G_k  + L_k x x^\top L_k^\top \right)B_{r, 11})\\
    &+ \log \det D + \log \det \Omega,
\end{aligned}
\end{equation*}
and therefore
\begin{equation*}
\begin{aligned}
    \mathbb{E}_{x \sim \mathcal{N}(0, C)} \mathbb{E}_{s \sim q_k(s;x)} & \left[ - 2\log p_\theta(x,s)\right]=\\
    &\tr(D^{-1}C) + \tr(D^{-1} F_1 B_{l, 11} (G_k + L_k C  L_k^\top) B_{r, 11} F_1^\top)\\
    &- 2\tr(D^{-1}C L_k^\top B_1^\top F_1^\top)\\
    &-2\tr(D^{-1} C L_k^\top B_2^\top F_2^\top) \\
    &+ 2 \tr(D^{-1}F_1  B_{l, 12}\left( G_k  + L_k C L_k^\top \right) B_{r, 12}F_2^\top) \\
    &+ \tr(D^{-1}F_2 B_{l, 22} \left( G_k  + L_k C L_k^\top \right)B_{r, 22} F_2^\top)\\
    &+\tr(\Omega^{-1} B_{l, 11} \left( G_k  + L_k C L_k^\top \right) B_{r, 11})\\
    &+ \log \det D + \log \det \Omega.
\end{aligned}
\end{equation*}
% At this point, we note that for any distribution $\mathbb{P}$ with zero mean and covariance $C$
% \begin{equation}
% \label{eq:any_distr}
%      \mathbb{E}_{x \sim \mathcal{N}(0, C)} \mathbb{E}_{q_k(s;x)}  \left[ - 2\log p_\theta(x,s)\right] =  \mathbb{E}_{x \sim \mathbb{P}} \mathbb{E}_{q_k(s;x)}  \left[ - 2\log p_\theta(x,s)\right].
% \end{equation}
% We will use this result later to show that the EM update is the same no matter the distribution of $x$, as long as this distribution is zero-mean with covariance $C$.
We now define
\begin{equation*}
\begin{aligned}
    C_{ss}^k &=: G_k + L_k C L_k^\top\\
    &=: \begin{bmatrix}
        C_{ss, 11}^k & C_{ss, 12}^k \\
        C_{ss, 12}^{k,\top} & C_{ss, 22}^k
    \end{bmatrix} =
     \begin{bmatrix}
        B_{l, 11} (G_k + L_k C  L_k^\top) B_{r, 11}  & B_{l, 12} (G_k + L_k C  L_k^\top) B_{r, 12} \\
        \left(B_{l, 12} (G_k + L_k C  L_k^\top) B_{r, 12} \right)^\top & B_{l, 22} (G_k + L_k C  L_k^\top) B_{r, 22}
    \end{bmatrix},
\end{aligned}
\end{equation*}
and partition
\begin{equation*}
    L_k = \begin{bmatrix}
        L_1^k \\
        L_2^k
    \end{bmatrix},
\end{equation*}
where $L_1^k \in \mathbb{R}^{n_1 \times n}$ and $L_2^k \in \mathbb{R}^{n_2 \times n}$.
We further define
    \begin{equation*}
        \Tilde{\Sigma}_k =: C + F_1 C_{ss,11}^k F_1^\top -2 C L_1^{k,\top} F_1^\top,\quad
        \Tilde{C}_{xs}^k =: C L_2^{k,\top} - F_1 C_{ss, 12}^k.
    \end{equation*}
Note that
\begin{equation*}
    \mathbb{E}_{x \sim \mathcal{N}(0, C)} \mathbb{E}_{q_k(s; x)} \left[ x s^\top \right] = \mathbb{E}_{x \sim \mathcal{N}(0, C)} \left[ xx^\top L_k^\top \right] = C L_k^\top.
\end{equation*}
We define $C_{xs}^k =: C L_k^\top = \begin{bmatrix}
    C L_1^{k, \top} & C L_2^{k, \top}
\end{bmatrix} = \begin{bmatrix}
    C_{xs, 1}^k & C_{xs, 2}^k
\end{bmatrix}$.

With these definitions,
\begin{equation}
\label{eq:final_obj}
\begin{aligned}
    \mathbb{E}_{x \sim \mathcal{N}(0, C)} \mathbb{E}_{q_k(s;x)} & \left[ - 2\log p_\theta(x,s)\right] = \\
    &\tr\left(D^{-1} \left( \Tilde{\Sigma}_k - 2\Tilde{C}_{xs}^kF_2^\top + F_2 C_{ss, 22}^k F_2^\top \right) \right)\\
    &+\tr(\Omega^{-1} C_{ss, 11}^k)\\
    &+\log \det \Omega + \log \det D.
\end{aligned}
\end{equation}
With our previous definitions, we can clearly see how each variable ($F_2, \Omega, D$) enters the optimization objective.

\subsection{Explicit Equations for the EM Update}
Using \eqref{eq:final_obj}, we are now ready to evaluate the EM update
\begin{equation*}
    \theta_{k+1}=\argmin_\theta \mathbb{E}_{x \sim \mathcal{N}(0, C)} \mathbb{E}_{s \sim q_k(s;x)}  \left[ - 2\log p_\theta(x,s)\right]
\end{equation*}
under the constraint that $D$ is diagonal. Notice that the problem is separable in $(D, F_2)$ and $\Omega$.

We do the easy part first and find the optimal value for $\Omega$. To do this, we take the derivative of
\[
\tr(\Omega^{-1} C_{ss, 11}^k)+\log \det \Omega
\]
with respect to $\Omega^{-1}$
and set it to zero. This gives:
\begin{equation*}
    \Omega^{\star} = C_{ss, 11}^{k}.
\end{equation*}
We now turn to the sub-problem that depends on $D, F$:
\begin{equation}
\label{eq:subprob}
    \minimize_{D} \minimize_{F_2} \tr\left( D^{-1} \left( \Tilde{\Sigma}_k - 2\Tilde{C}_{xs}^kF_2^\top + F_2 C_{ss, 22}^k F_2^\top \right) \right) + \log \det D.
\end{equation}
We fix $D$. The optimal $F_2$ for this $D$ can be found by taking the derivative of the objective with respect to $F_2$ and setting it to zero. This gives
\begin{align*}
    D^{-1} (-2 \Tilde{C}_{xs}^k + 2 F_2^{\star}(D) C_{ss, 22}^k) = 0 \Leftrightarrow F_2^\star(D) = \Tilde{C}_{xs}^k C_{ss, 22}^{k, -1}.
\end{align*}
It only remains to minimize over the diagonal $D$. We plug $F_2^\star(D)$ in the objective of \eqref{eq:subprob}, which takes care of the internal minimization with respect to $F_2$. To obtain the optimal diagonal $D$, we solve
\begin{equation}
\label{eq:D}
    \minimize_{D} \tr\left(D^{-1} \left(\Tilde{\Sigma}_k - \Tilde{C}_{xs}^k C_{ss, 22}^{k, -1} \Tilde{C}_{xs}^{k, \top} \right)\right) + \log \det D.
\end{equation}
Let $D_{ii}=1/\psi_i$. Then, problem \eqref{eq:D} is equivalent to
\begin{equation*}
    \minimize_{\psi_1, \dots, \psi_n} \sum_{i=1}^n \psi_i \left(\Tilde{\Sigma}_k - \Tilde{C}_{xs}^k C_{ss, 22}^{k, -1} \Tilde{C}_{xs}^{k, \top} \right)_{ii} - \log \psi_i.
\end{equation*}
The last problem is separable in the $\psi_i$'s. By taking the derivatives with respect to the $\psi_i$'s and setting them to zero, we get the solution
\begin{equation*}
    \dfrac{1}{\psi_i^\star} = \left(\Tilde{\Sigma}_k - \Tilde{C}_{xs}^k C_{ss, 22}^{k, -1} \Tilde{C}_{xs}^{k, \top} \right)_{ii}.
\end{equation*}
Therefore, it holds that
\begin{equation*}
    D^\star = \diag \left(\Tilde{\Sigma}_k - \Tilde{C}_{xs}^k C_{ss, 22}^{k, -1} \Tilde{C}_{xs}^{k, \top} \right).
\end{equation*}
To summarize, the EM update is
\begin{equation}
\label{eq:updatesss}
\begin{aligned}
    \Omega &\leftarrow C_{ss, 11}^k\\
    F_2 &\leftarrow \Tilde{C}_{xs}^k C_{ss, 22}^{k, -1}\\
    D &\leftarrow \diag \left(\Tilde{\Sigma}_k - \Tilde{C}_{xs}^k C_{ss, 22}^{k, -1} \Tilde{C}_{xs}^{k, \top} \right).
\end{aligned}
\end{equation}
We expect that the diagonal of the obtained risk model $\tF_k \tOmega_k \tF_k^\top +D_k$ will be \rf{equal} to the diagonal of $C$. When fitting a low-rank-plus-diagonal model by maximizing the Gaussian log likelihood, this is implied by the first-order optimality condition with respect to $D$ \cite[Chapter 9]{mardia2024multivariate}.

The update \eqref{eq:updatesss} can equivalently be written as
\begin{equation*}
\begin{aligned}
    \Omega &\leftarrow C_{ss, 11}^k\\
    F_2 &\leftarrow \left( C_{xs, 2}^k - F_1 C_{ss, 12}^k \right) C_{ss, 22}^{k, -1}\\
    F & \leftarrow \begin{bmatrix}
        F_1 & F_2
    \end{bmatrix}\\
    D &\leftarrow \diag \left(C - 2 C_{xs}^k F^\top + F C_{ss}^k F^\top  \right).
\end{aligned}
\end{equation*}
Here, we used the fact that $\mathrm{diag}(F_1 C_{ss, 12}^k F_2^\top) = \mathrm{diag}(F_2 C_{ss, 21}^k F_1^\top)$, because the diagonal of a square matrix is equal to the diagonal of its transpose.

% \subsubsection{Interpretation as maximum likelihood estimation}
% By \eqref{eq:any_distr}, for any distribution $\mathbb{P}$ with zero mean and covariance $C$, the EM update satisfies
% \begin{equation}
%     \argmin_\theta \mathbb{E}_{x \sim \mathcal{N}(0, C)} \mathbb{E}_{q_k(s;x)}  \left[ - 2\log p_\theta(x,s)\right] = \argmin_\theta \mathbb{E}_{x \sim \mathbb{P}} \mathbb{E}_{q_k(s;x)}  \left[ - 2\log p_\theta(x,s)\right].
% \end{equation}
% This implies that minimizing the KL divergence between $\mathbb{P}$ and $\mathcal{N}(0, \tF \tOmega \tF^\top +D)$ with EM is equivalent to minimizing the KL divergence between $\mathcal{N}(0, C)$ and $\mathcal{N}(0, \tF \tOmega \tF^\top +D)$ with EM. The update only depends on the covariance $C$, which is common for $\mathbb{P}$ and $\mathcal{N}(0, C)$. This justifies calling this method maximum likelihood covariance estimation, as we are fitting a Gaussian to empirical data.

% When the covariance $C$ is not available, we can replace it with any empirical covariance estimate, such as a EWMA, from the observed samples.

\subsection{Evidence Lower Bound (ELBO) for EM}
\label{appendix:elbo}
The EM algorithm maximizes a sequence of lower bounds, as seen by
\begin{equation*}
    \maximize_\theta \mathbb{E}_{x \sim \mathcal{N}(0, C)} \log p_\theta(x) \geq \maximize_\theta \mathbb{E}_{x \sim \mathcal{N}(0, C)} \mathbb{E}_{s \sim q(s;x)} \log \dfrac{p_\theta(x,s)}{q(s;x)}
\end{equation*}
for any $q(s;x)$.
If we set $q(s;x)$ to be the conditional distribution $p_\theta(s\mid x)$, then
\begin{equation}
\label{eq:eqcond}
    \log \dfrac{p_\theta(x, s)}{q(s;x)} = \log p_\theta(x) \Rightarrow \mathbb{E}_{q(s;x)} \log \dfrac{p_\theta(x,s)}{q(s;x)} = \log p_\theta(x).
\end{equation}
Let $J(\theta) =  \mathbb{E}_{x \sim \mathcal{N}(0, C)} \log p_\theta(x)$. Then, by eq.~\eqref{eq:ineq},
\begin{align*}
    J(\theta_{k+1}) \geq \mathbb{E}_{x \sim \mathcal{N}(0, C)} \mathbb{E}_{s \sim q_k(s; x)} \log \dfrac{p_{\theta_{k+1}}(x,s)}{q_k(s;x)}
\end{align*}
and by \eqref{eq:emem},
\begin{equation*}
    \mathbb{E}_{x \sim \mathcal{N}(0, C)} \mathbb{E}_{s \sim q_k(s; x)} \log \dfrac{p_{\theta_{k+1}}(x,s)}{q_k(s;x)} \geq \mathbb{E}_{x \sim \mathcal{N}(0, C)} \mathbb{E}_{s \sim q_k(s; x)} \log \dfrac{p_{\theta_{k}}(x,s)}{q_k(s;x)}.
\end{equation*}
By \eqref{eq:eqcond},
\begin{equation*}
    \mathbb{E}_{x \sim \mathcal{N}(0, C)} \mathbb{E}_{s \sim q_k(s; x)} \log \dfrac{p_{\theta_{k}}(x,s)}{q_k(s;x)} = \mathbb{E}_{x \sim \mathcal{N}(0, C)} \log p_{\theta_k}(x) = J(\theta_k).
\end{equation*}
Combining the above results,
\begin{equation*}
    J(\theta_{k+1}) \geq J(\theta_k)
\end{equation*}
for the EM update. In other words, EM improves the expected log-likelihood at every iteration.

\section{The EM Update with Missing Returns}
\label{appendix:em_mis}
\subsection{The Setup}
We derive the EM algorithm in the case of missing returns data. We consider a time window of length $T$. Let $X_\tau$ be the returns random
vector at time $\tau=1, \dots, T$. We assume that at time $\tau$ we only observe the returns $X_\tau$ on a subset of assets indexed by the set $\mathcal{O}_\tau \subseteq \{1, \dots, n\}$. Let $\mathcal{M}_\tau = \{1, \dots, n\} \setminus \mathcal{O}_\tau$ be the set of assets with missing returns at time $\tau$. We denote the observed portion of the return at $\tau$ as $X_{\tau, \mathrm{obs}}$ (with value $x_{\tau, \mathrm{obs}}$) and the missing portion as $X_{\tau, \mathrm{mis}}$.
We define the permutation matrix $\Pi_\tau \in \mathbb{R}^{n \times n}$ that satisfies
\begin{equation*}
    \begin{bmatrix}
        X_{\tau, \mathrm{mis}} \\
        X_{\tau, \mathrm{obs}}
    \end{bmatrix} = \Pi_\tau X_\tau.
\end{equation*}
We wish to solve
\begin{equation}
    \label{eq:main_mis}
    \maximize_\theta \sum_{\tau \leq T} w_\tau \log p_{\theta, \mathrm{obs}, \tau}(x_{\tau, \mathrm{obs}}),
\end{equation}
where $p_\theta(x)$ is the probability density function of the Gaussian $\mathcal{N}(0, \tF \tOmega \tF^\top +D)$,
\begin{equation*}
    \tF = \begin{bmatrix}
        F_1 & F_2
    \end{bmatrix},\qquad \tOmega = \begin{bmatrix}
        \Omega & 0 \\
        0 & I
    \end{bmatrix},
\end{equation*}
$p_{\theta, \mathrm{obs}, \tau}(\cdot)$ is the marginal density of the observed returns at time $\tau$,
$\theta=:(F_2, \Omega, D)$, and $w_\tau$ are non-negative weights that sum to $1$.

It holds that
\begin{equation*}
    \begin{aligned}
    \log p_{\theta, \mathrm{obs}, \tau}(x_{\tau, \mathrm{obs}}) = \log \int p_\theta(x_\tau, s_\tau) dx_{\tau, \mathrm{mis}} ds_\tau = \\
    \log \int q(s_\tau, x_{\tau, \mathrm{mis}}; x_{\tau, \mathrm{obs}}) \dfrac{p_\theta(x_\tau, s_\tau) }{q(s_\tau, x_{\tau, \mathrm{mis}}; x_{\tau, \mathrm{obs}})} dx_{\tau, \mathrm{mis}} ds_\tau \geq \\
    \mathbb{E}_{s_\tau,x_{\tau, \mathrm{mis}} \sim q(s_\tau, x_{\tau, \mathrm{mis}}; x_{\tau, \mathrm{obs}})} \log \dfrac{p_\theta(x_\tau, s_\tau)}{q(s_\tau, x_{\tau, \mathrm{mis}}; x_{\tau, \mathrm{obs}})},
    \end{aligned}
\end{equation*}
where $x_\tau = \Pi_\tau^\top \begin{bmatrix}
    x_{\tau, \mathrm{mis}} \\
    x_{\tau, \mathrm{obs}}
\end{bmatrix}$, and $q(s_\tau, x_{\tau, \mathrm{mis}}; x_{\tau, \mathrm{obs}})$ is any probability density function over the latent variables $s_\tau$ and the missing returns $x_{\tau, \mathrm{mis}}$, parameterized by the observed returns $x_{\tau, \mathrm{obs}}$. The inequality follows from Jensen's inequality. The analysis above is similar to the analysis of \Cref{appendix:em}. Recall that $p_\theta(x,s)$ is the joint density of $(X,S)$ evaluated at $X=x$ and $S=s$.

We now have a lower bound for each term in the objective of \eqref{eq:main_mis}. Instead of solving \eqref{eq:main_mis} directly, we can maximize the lower bound
\begin{equation*}
    \maximize_\theta \sum_{\tau \leq T} w_\tau \mathbb{E}_{s_\tau,x_{\tau, \mathrm{mis}} \sim q(s_\tau, x_{\tau, \mathrm{mis}}; x_{\tau, \mathrm{obs}})} \log \dfrac{p_\theta(x_\tau, s_\tau)}{q(s_\tau, x_{\tau, \mathrm{mis}}; x_{\tau, \mathrm{obs}})}
\end{equation*}
that has the same solution as
\begin{equation}
    \label{eq:mis_main_lb}
    \minimize_\theta \sum_{\tau \leq T} w_\tau \mathbb{E}_{s_\tau,x_{\tau, \mathrm{mis}} \sim q(s_\tau, x_{\tau, \mathrm{mis}}; x_{\tau, \mathrm{obs}})} \left[ - 2\log p_\theta(x_\tau, s_\tau) \right].
\end{equation}
Note that as derived in \Cref{appendix:em},
\begin{equation*}
\begin{aligned}
    \mathbb{E}_{q_\tau} &\left[ - 2\log p_\theta(x_\tau,s_\tau)\right] = \\
    &\tr(D^{-1} \mathbb{E}_{q_\tau} \left[ x_\tau x_\tau^\top \right]) + \tr(D^{-1} F_1 \mathbb{E}_{q_\tau} \left[ s_{\tau,1} s_{\tau, 1}^\top \right] F_1^\top) \\
    &- 2\tr(D^{-1} \mathbb{E}_{q_\tau}\left[ x_\tau s_{\tau, 1}^\top\right]F_1^\top)\\
    &-2\tr(D^{-1} \mathbb{E}_{q_\tau} \left[ x_\tau s_{\tau, 2}^\top  \right] F_2^\top) + 2 \tr(D^{-1}F_1 \mathbb{E}_{q_\tau}\left[ s_{\tau,1} s_{\tau,2}^\top \right] F_2^\top) \\
    &+ \tr(D^{-1}F_2 \mathbb{E}_{q_\tau}\left[ s_{\tau, 2} s_{\tau, 2}^\top\right] F_2^\top)\\
    &+\tr(\Omega^{-1} \mathbb{E}_{q_\tau}\left[ s_{\tau,1} s_{\tau,1}^\top \right])\\
    &+ \log \det D + \log \det \Omega,
\end{aligned}
\end{equation*}
where we abbreviate $s_\tau,x_{\tau, \mathrm{mis}} \sim q(s_\tau, x_{\tau, \mathrm{mis}}; x_{\tau, \mathrm{obs}})$ as $q_\tau$ for simplicity.

EM solves \eqref{eq:mis_main_lb} iteratively. Assuming we are at iteration $k$, with $\theta_k$, we set
\begin{equation}
    \label{eq:q_mis}
    q(s_\tau, x_{\tau, \mathrm{mis}}; x_{\tau, \mathrm{obs}}) = p_{\theta_k}(S_\tau=s_\tau, X_{\tau, \mathrm{mis}}=x_{\tau, \mathrm{mis}} \mid X_{\tau, \mathrm{obs}} = x_{\tau, \mathrm{obs}}).
\end{equation}
We remind the reader that for arbitrary $\theta$, the random variables satisfy
\begin{equation*}
    \begin{aligned}
        X_{\tau, \mathrm{obs}} = \tF_{\tau, \mathrm{obs}} S_\tau + E_{\tau, \mathrm{obs}} \\
        X_{\tau, \mathrm{mis}} = \tF_{\tau, \mathrm{mis}} S_\tau + E_{\tau, \mathrm{mis}},
    \end{aligned}
\end{equation*}
where
\begin{equation*}
    S_\tau =\begin{bmatrix}
        S_{\tau,1} \\
        S_{\tau,2}
    \end{bmatrix}\sim \mathcal{N}(0, \begin{bmatrix}
    \Omega & 0\\
    0 & I
\end{bmatrix}), \quad E_\tau \sim \mathcal{N}(0, D),
\end{equation*}
and $S_\tau$ is independent of $E_\tau$. Also, to make the notation clear: $\tF_{\tau, \mathrm{obs}}$ refers to the rows of $\tF$ indexed by $\mathcal{O}_\tau$, and similarly for the other quantities indexed by `obs' or `mis'.

These relations allow us to compute the conditional distribution in \eqref{eq:q_mis} and the necessary statistical moments, as we show next.

\subsection{Computing the Statistical Moments}
Similar to \Cref{appendix:em}, the conditional distribution of $S_\tau$ given the observed returns $X_{\tau, \mathrm{obs}} = x_{\tau, \mathrm{obs}}$ at iteration $k$ (i.e., with $\tF_k, \Omega_k, D_k$) is
\begin{equation*}
    S_\tau \mid X_{\tau, \mathrm{obs}} = x_{\tau, \mathrm{obs}}; \theta_k \sim \mathcal{N}(L_{k,\tau} x_{\tau, \mathrm{obs}}, G_{k,\tau}),
\end{equation*}
where
\begin{equation*}
    L_{k,\tau} = G_{k, \tau} \tF_{k, \tau, \mathrm{obs}}^\top D_{k, \tau, \mathrm{obs}}^{-1}, \quad G_{k,\tau}^{-1} = \tF_{k,\tau, \mathrm{obs}}^\top D_{k, \tau, \mathrm{obs}}^{-1} \tF_{k, \tau, \mathrm{obs}} + \tOmega_{k}^{-1}.
\end{equation*}
To make the notation clear, $\tF_{k, \tau, \mathrm{obs}}$ is the sub-matrix of $\tF_k$ that contains the rows indexed by $\mathcal{O}_\tau$ and $D_{k, \tau, \mathrm{obs}}$ is the sub-matrix of $D_k$ that contains the rows and columns indexed by $\mathcal{O}_\tau$.
We also define
\begin{equation*}
    \hat{s}_\tau =: L_{k, \tau} x_{\tau, \mathrm{obs}}.
\end{equation*}
We can now compute the necessary statistical moments at iteration $k$. First, similar to \Cref{appendix:em}, we get
\begin{equation*}
    \mathbb{E} \left[ S_\tau S_\tau^\top \mid X_{\tau, \mathrm{obs}} = x_{\tau, \mathrm{obs}}; \theta_k\right] = G_{k,\tau} + \hat{s}_\tau \hat{s}_\tau^\top.
\end{equation*}
Second,
\begin{equation}
    \label{eq:moment_obs}
    \mathbb{E} \left[ X_{\tau, \mathrm{obs}} S_\tau^\top \mid X_{\tau, \mathrm{obs}} = x_{\tau, \mathrm{obs}}; \theta_k \right] = x_{\tau, \mathrm{obs}} \hat{s}_\tau^\top.
\end{equation}
Third,
\begin{equation}
    \label{eq:moment_mis}
    \begin{aligned}
    \mathbb{E} \left[ X_{\tau, \mathrm{mis}} S_\tau^\top \mid X_{\tau, \mathrm{obs}} = x_{\tau, \mathrm{obs}}; \theta_k \right] = \tF_{k, \tau, \mathrm{mis}} \mathbb{E} \left[ S_\tau S_\tau^\top \mid X_{\tau, \mathrm{obs}} = x_{\tau, \mathrm{obs}}; \theta_k\right] \\
    =\tF_{k, \tau, \mathrm{mis}} (\hat{s}_\tau \hat{s}_\tau^\top + G_{k,\tau}) = \hat{x}_{\tau, \mathrm{mis}} \hat{s}_\tau^\top + \tF_{k, \tau, \mathrm{mis}} G_{k,\tau},
    \end{aligned}
\end{equation}
where $\hat{x}_{\tau, \mathrm{mis}} = \tF_{k, \tau, \mathrm{mis}} \hat{s}_\tau$ are the imputed missing values.
Here, we leveraged the fact that $E_{\tau, \mathrm{mis}}$ is independent of $(E_{\tau, \mathrm{obs}}, S_\tau)$.
Equations~\eqref{eq:moment_obs}--\eqref{eq:moment_mis} imply that
\begin{equation*}
\begin{aligned}
    \mathbb{E} \left[ X_\tau S_\tau^\top \mid X_{\tau, \mathrm{obs}} = x_{\tau, \mathrm{obs}}; \theta_k \right] &=\mathbb{E} \left[  \Pi_\tau^\top \begin{bmatrix}
        X_{\tau, \mathrm{mis}} \\
        X_{\tau, \mathrm{obs}}
    \end{bmatrix} S_\tau^\top \mid X_{\tau, \mathrm{obs}} = x_{\tau, \mathrm{obs}}; \theta_k \right] \\
    &=\Pi_\tau^\top \begin{bmatrix}
        \hat{x}_{\tau, \mathrm{mis}}  \\
        x_{\tau, \mathrm{obs}}
    \end{bmatrix} \hat{s}_\tau^\top + \Pi_\tau^\top  \begin{bmatrix}
        \tF_{k, \tau, \mathrm{mis}} G_{k,\tau} \\
        0
    \end{bmatrix},
\end{aligned}
\end{equation*}
where
\[
\hat{x}_\tau = \Pi_\tau^\top \begin{bmatrix}
        \hat{x}_{\tau, \mathrm{mis}}  \\
        x_{\tau, \mathrm{obs}}
    \end{bmatrix}
\]
are the imputed values for $X_\tau$.
Finally, it holds that
\[
\mathbb{E} \left[ X_\tau X_\tau^\top \mid X_{\tau, \mathrm{obs}} = x_{\tau, \mathrm{obs}}; \theta_k \right] = \Pi_\tau^\top \mathbb{E} \left[ \begin{bmatrix}
    X_{\tau, \mathrm{mis}} \\
    X_{\tau, \mathrm{obs}}
\end{bmatrix} \begin{bmatrix}
    X_{\tau, \mathrm{mis}} \\
    X_{\tau, \mathrm{obs}}
\end{bmatrix}^\top \mid X_{\tau, \mathrm{obs}} = x_{\tau, \mathrm{obs}}; \theta_k \right]  \Pi_\tau,
\]
where, because $E_{\tau, \mathrm{mis}}$ is independent of $(E_{\tau, \mathrm{obs}}, S_\tau)$, we can write
\[
\begin{aligned}
\mathbb{E} \left[ X_{\tau, \mathrm{mis}} X_{\tau, \mathrm{mis}}^\top \mid X_{\tau, \mathrm{obs}} = x_{\tau, \mathrm{obs}}; \theta_k  \right]
&=
\Tilde{F}_{k, \tau, \mathrm{mis}} \mathbb{E} \left[ S_\tau S_\tau^\top  \mid X_{\tau, \mathrm{obs}} = x_{\tau, \mathrm{obs}}; \theta_k  \right] \Tilde{F}_{k, \tau, \mathrm{mis}}^\top \\
&+ \mathbb{E} \left[ E_{\tau, \mathrm{mis}} E_{\tau, \mathrm{mis}}^\top \right] \\
&= \hat{x}_{\tau, \mathrm{mis}} \hat{x}_{\tau, \mathrm{mis}}^\top + \tF_{k, \tau, \mathrm{mis}} G_{k,\tau} \tF_{k, \tau, \mathrm{mis}}^\top + D_{k, \tau, \mathrm{mis}}.
\end{aligned}
\]
Furthermore,
\[
\begin{aligned}
\mathbb{E} \left[ X_{\tau, \mathrm{obs}} X_{\tau, \mathrm{mis}}^\top \mid X_{\tau, \mathrm{obs}} = x_{\tau, \mathrm{obs}}; \theta_k  \right] &= \mathbb{E} \left[ X_{\tau, \mathrm{obs}} S_\tau^\top \mid X_{\tau, \mathrm{obs}} = x_{\tau, \mathrm{obs}}; \theta_k \right] \Tilde{F}_{k, \tau, \mathrm{mis}}^\top \\
&= {x}_{\tau, \mathrm{obs}} \hat{x}_{\tau, \mathrm{mis}}^\top,
\end{aligned}
\]
\[
\mathbb{E} \left[ X_{\tau, \mathrm{obs}} X_{\tau, \mathrm{obs}}^\top \mid X_{\tau, \mathrm{obs}} = x_{\tau, \mathrm{obs}}; \theta_k  \right] = x_{\tau, \mathrm{obs}} x_{\tau, \mathrm{obs}}^\top.
\]
We can combine the expressions above to write
\begin{equation*}
    \mathbb{E} \left[ X_\tau X_\tau^\top \mid X_{\tau, \mathrm{obs}} = x_{\tau, \mathrm{obs}}; \theta_k \right] = \hat{x}_\tau \hat{x}_\tau^\top + \Pi_\tau^\top \begin{bmatrix}
        \tF_{k, \tau, \mathrm{mis}} G_{k,\tau} \tF_{k, \tau, \mathrm{mis}}^\top +D_{k, \tau, \mathrm{mis}} & 0 \\
        0 & 0
    \end{bmatrix} \Pi_\tau.
\end{equation*}
The zero entries of the matrix in the right-hand side are reasonable, because the conditional covariance of any quantity involving $X_{\tau, \mathrm{obs}}$ is $0$.

\subsection{Explicit Equations for the EM Update}
We are at iteration $k$, with $\theta_k$. Let
\begin{itemize}
    \item $C_{xx}^k =: \sum_{\tau \leq T} w_\tau \mathbb{E} \left[ X_\tau X_\tau^\top \mid X_{\tau, \mathrm{obs}} = x_{\tau, \mathrm{obs}}; \theta_k \right]$,
    \item $W_{xs}^k =: \sum_{\tau \leq T} w_\tau \mathbb{E} \left[ X_\tau S_\tau^\top \mid X_{\tau, \mathrm{obs}} = x_{\tau, \mathrm{obs}}; \theta_k \right] = \begin{bmatrix}
        W_{xs, 1}^k &  W_{xs, 2}^k
    \end{bmatrix}$,
    \item $C_{ss}^k =: \sum_{\tau \leq T} w_\tau \mathbb{E} \left[ S_\tau S_\tau^\top \mid X_{\tau, \mathrm{obs}} = x_{\tau, \mathrm{obs}}; \theta_k \right] = \begin{bmatrix}
        C_{ss, 11}^k & C_{ss, 12}^k\\
        C_{ss, 21}^k & C_{ss, 22}^k
    \end{bmatrix}$.
\end{itemize}
These can be computed as shown at the end of the previous subsection.
EM solves problem \eqref{eq:mis_main_lb} using \eqref{eq:q_mis} to get $\theta_{k+1}$. With our new definitions, we can write problem \eqref{eq:mis_main_lb} as
\begin{equation}
    \label{eq:equiv_mis}
    \begin{aligned}
    \minimize_\theta &\tr\Bigl(D^{-1} \Bigl( C_{xx}^k + F_1 C_{ss, 11}^k F_1^\top - 2 W_{xs, 1}^k F_1^\top \\
    &\qquad - 2 W_{xs, 2}^k F_2^\top + 2 F_1 C_{ss, 12}^k F_2^\top + F_2 C_{ss, 22}^k F_2^\top \Bigr) \Bigr) \\
    &+ \log \det D \\
    &+ \tr(\Omega^{-1} C_{ss, 11}^k) + \log \det \Omega.
    \end{aligned}
\end{equation}
We further define
\begin{itemize}
    \item $\Tilde{\Sigma}_k =: C_{xx}^k + F_1 C_{ss,11}^k F_1^\top -2 W_{xs,1}^k F_1^\top$,
    \item $\Tilde{C}_{xs}^k =: W_{xs,2}^k - F_1 C_{ss, 12}^k$.
\end{itemize}
With these definitions, \eqref{eq:equiv_mis} becomes
\begin{equation*}
\begin{aligned}
    \minimize_\theta
    &\tr\left(D^{-1} \left( \Tilde{\Sigma}_k - 2\Tilde{C}_{xs}^kF_2^\top + F_2 C_{ss, 22}^k F_2^\top \right) \right)\\
    & + \log \det D\\
    &+\tr(\Omega^{-1} C_{ss, 11}^k) +\log \det \Omega.
\end{aligned}
\end{equation*}
We have already solved this problem in \Cref{appendix:em}. Therefore, the EM update with missing returns data is

\begin{equation*}
\begin{aligned}
    \Omega &\leftarrow C_{ss, 11}^k\\
    F_2 &\leftarrow \Tilde{C}_{xs}^k C_{ss, 22}^{k, -1}\\
    D &\leftarrow \diag \left(\Tilde{\Sigma}_k - \Tilde{C}_{xs}^k C_{ss, 22}^{k, -1} \Tilde{C}_{xs}^{k, \top} \right).
\end{aligned}
\end{equation*}

\section{Additional Analyses}
\label{appendix:additional-analyses}

We present additional analyses that complement the main text. 
We first quantify how well our method recovers the covariance when the covariance is known. 
We then test how the enhancing methodology performs out-of-sample, when we know the model that generates returns and still enhance it with additional
statistical factors. 
We also construct a daily Fama--French base model and extend it with ten transient factors, one at a time, and attempt an 
interpretation of the added factors in terms of sector exposures.
Finally, we compare our extended model with the Barra model when both have 73 factors. 
To do this, we randomly remove two style factors and five industry factors from the full Barra model and apply our method to learn seven additional factors.

\subsection{Covariance Recovery}\label{sec:e1}
We first quantify how well our method estimates the covariance, when the covariance is known.

We draw $F\in\mathbb R^{100\times15}$ with independent standard Gaussian
entries, center each column across assets, and divide it by its standard
deviation.
Independently, we draw $d_i\sim U[0,1]+0.5$ and
set $D=\diag(d)$ and $\Sigma=FF^\top+D$. This $F,D$ stays fixed in \cref{sec:e1} and
\cref{sec:e2}. 

Each of 10 independent runs
draws $T=1000$ returns from $\mathcal N(0,\Sigma)$ and fits EM (\cref{alg:em}) with
$n_1=0,n_2=15$ and equal weights, using all $T$ observations.
We compare $\hat F\hat F^\top$ with $FF^\top$, since columns can rotate
without changing the covariance.

The empirical covariance is $C=T^{-1}\sum_\tau x_\tau x_\tau^\top$, where $x_\tau$ are the sampled returns.
The known mean is zero: returns are neither demeaned nor standardized.
We compare Algorithm~\ref{alg:em} with the following method that also estimates $\Sigma, F$, and $D$, which we call `PCA+diagonal'. Our comparator retains the top 15 eigenpairs of $C$, which we write as
$G_{15}=U_{15}\Lambda_{15}^{1/2}$, and sets
\[
 \hat D_{\mathrm{PCA}}=\diag\!\left((\diag(C-G_{15}G_{15}^\top))_+\right).
\]

\begin{sloppypar}
Table~\ref{tab:additional-recovery} reports means and sample standard deviations across the 10
independent fitting samples. For a component $M$, relative error means
$\|\hat M-M\|_F/\|M\|_F$; diagonal recovery is $\tr(\hat D)/\tr(D)$.

\end{sloppypar}

\begin{table}[h!]
\centering
\small
\setlength{\tabcolsep}{3pt}
\begin{tabular}{lrrrr}
\toprule
Method & \makecell{Error in\\$\Sigma$} & \makecell{Error in\\$FF^\top$} & \makecell{Error in\\$D$} & \makecell{Diagonal\\recovery} \\
\midrule
EM & $12.57 \pm 0.99$ & $12.69 \pm 1.00$ & $5.23 \pm 0.29$ & $98.67 \pm 0.64$ \\
PCA + diagonal & $12.63 \pm 0.99$ & $12.75 \pm 1.00$ & $18.72 \pm 0.55$ & $83.81 \pm 0.53$ \\
Empirical $C$ & $12.59 \pm 0.99$ & --- & --- & --- \\
\bottomrule
\end{tabular}
\caption{Recovery from 1000 observations with $d_i\sim U[0,1]+0.5$.
All entries are percentages.}
\label{tab:additional-recovery}
\end{table}

EM substantially improves recovery of $D$, with a small downward bias
(1.33\% in total diagonal variance). The approximately 12.57\%
total covariance error is consistent with sampling variation. Although both methods are good at estimating $\Sigma, FF^\top$, EM is also good at estimating the idiosyncratic component of the covariance.

We also fit the exact population covariance $C=\Sigma$ using \cref{alg:em}. EM recovers total
and common covariance to relative error below $5\times10^{-9}$ and $D$
below $10^{-12}$. PCA still has 17.34\% relative error in $D$ and retains
only 84.93\% of its total variance.

\subsection{Extending the Model that Generates the Returns}\label{sec:e2}
We test how the enhancing methodology performs out-of-sample, when we know the true model that generares returns. We assume the true model is $\mathcal{N}(0, \Sigma)$ and we know $F, D$.
We extend the true covariance with $5$ additional factors using sampled returns from the true model and judge the out-of-sample performance as done in section~\ref{sec:r2_ret}.

We perform $J=50$ independent fits using the same fixed
$\Sigma=FF^\top+D$,
and $T=1000$ new observations per fit. The provided base model is
exact: $F_1=F$, $\Omega_{\mathrm{base}}=I$, and $D_{\mathrm{base}}=D$.
We add $n_2=5$ factors, so the extended model has 20 factors.
Algorithm~\ref{alg:em} re-estimates $\Omega,D$ and learns $F_2$.

For each fitted model $j$, draw $K=10$ additional independent
returns from $\mathcal{N}(0, \Sigma)$. For each return $k$, draw one fresh 90/10 asset split and apply
equation~\eqref{eq:ltsq_r^2} and the subsequent prediction and $R^2$ definitions in \Cref{sec:r2_ret} to the true and extended models using the \emph{same}
return and split. The extended model remains frozen throughout these
$K$ comparisons. Define
\[
 \Delta_{jk}=R^2_{jk,\mathrm{true}}-R^2_{jk,\mathrm{extended}},\qquad
 \bar\Delta=\frac1{JK}\sum_{j,k}\Delta_{jk}.
\]
\begin{center}
\small
\begin{tabular}{lr}
\toprule
Mean $R^2$, true model & 0.89 \\%89391 \\
Mean $R^2$, extended model & 0.85\\ %7650 \\
Mean paired difference $\bar\Delta$ & 0.042\\ %1741 \\
Requested $\operatorname{std}(\Delta)/\sqrt{500}$ & 0.007\\ %6966 \\
\bottomrule
\end{tabular}
\end{center}
The results suggest that extending the true covariance does not lead to better out-of-sample $R^2$. 
However, following the initialization of \cref{alg:em}, the final training log likelihood was
lower than that of the base model in 41 of the 50 runs, so the observed
out-of-sample deterioration could reflect both statistical estimation
error and suboptimal optimization. We therefore repeated the same 50
fits by using $\Omega^{(1)}\leftarrow \Omega_{\mathrm{base}}$,
$D^{(1)} \leftarrow D_{\mathrm{base}}$, and
$F_2^{(1)} \leftarrow LU_5(\Lambda_5-I)_+^{1/2}$ in the initialization of \cref{alg:em} instead, where
$\Sigma_{\mathrm{base}}=F F^T + D=LL^\top$ and $U_5,\Lambda_5$ contain the leading
five eigenvectors and eigenvalues of $L^{-1}CL^{-\top}$, with $C$ the
training second-moment matrix. This initialization adds covariance in
directions where the sample variance exceeds that predicted by the base
model: $L^{-1}CL^{-\top}$ should be close to $I$. 
The final training log likelihood then exceeded the base model's
in all 50 runs. Nevertheless, using the same 500 out-of-sample returns
and paired asset splits, the true model retained an advantage:
$\bar\Delta=0.0013$ and
$\operatorname{std}(\{\Delta_{jk}\})/\sqrt{500}=0.0004$.

\subsection{Adding 10 Transient Factors Sequentially}

We construct a daily Fama--French base model and extend it with ten
transient factors, one at a time. The universe contains 483 securities
(480 companies) from the S\&P~500 constituent snapshot of
November 5, 2026. Each has complete split- and dividend-adjusted daily
returns from Yahoo Finance on all 1,733 trading days from
2020-01-02 to 2026-11-23. We exclude 20 securities with
incomplete histories and do not fill missing observations. The seven
factor returns $s_t$ are the Fama--French five factors~\cite{famaFrench2015},
momentum, and short-term reversal from Kenneth French's
library~\cite{frenchLibrary}, converted to decimal returns. We use raw
stock returns $x_t$, without subtracting the risk-free rate.

\paragraph{The base model.}
Let $I_t$ contain the preceding 252 trading days, ending at $t-1$.
We estimate the exposures by regression without an intercept,
\[
 F_t=\left(\sum_{\tau\in I_t}x_\tau s_\tau^\top\right)
       \left(\sum_{\tau\in I_t}s_\tau s_\tau^\top\right)^{-1}.
\]
We use the centered factor covariance
$\Omega_t=252^{-1}\sum_{\tau\in I_t}(s_\tau-\bar s_t)(s_\tau-\bar s_t)^\top$,
the variance estimates $v_t=252^{-1}\sum_{\tau\in I_t} x_\tau^2$ (elementwise square),
and form the stock return covariance estimate
\[
 D_t=\diag\!\left((v_t-\diag(F_t\Omega_tF_t^\top))_+\right),
 \qquad \Sigma_t=F_t\Omega_tF_t^\top+D_t.
\]

\paragraph{Sequential enhancement with warm starts.}
We attempt to tie the added statistical factors to economic and sector themes.
We do this by adding one factor at a time using the enhancing methodology.
Economic and sector themes should persist across adjacent trading days.
We therefore initialize each new factor with yesterday's loading at the
same stage, allowing the likelihood fit to adjust that loading.

On each date $t$ of the 689 model dates from 2022-07-01 to 2025-03-31,
initialize the base representation as
\[
 A_{t,0}=F_t\Omega_t^{1/2},\qquad D_{t,0}=D_t.
\]
Here $M^{1/2}$ denotes a Cholesky factor satisfying
$M^{1/2}(M^{1/2})^\top=M$. At time $t$ and stage $k=0,\ldots,9$, run
Algorithm~\ref{alg:em} on the base model $A_{t,k}A_{t,k}^\top+D_{t,k}$ with
the intitialization
\[
 n_2=1,\qquad F_1\leftarrow A_{t,k},\qquad
 \Omega_{\mathrm{base}}=I,\qquad F_2^{(0)}\leftarrow f_{t-1,k},
\]
where $f_{t-1,k}$ is the last column of $A_{t-1,k+1}$.

We run Algorithm~\ref{alg:em}
and retain its output.

From the returned $(\widehat F_{2,t,k},\widehat\Omega_{t,k},\widehat D_{t,k})$, form
\[
 \widetilde F_{t,k}=\big[A_{t,k},\widehat F_{2,t,k}\big],\qquad
 \widetilde\Omega_{t,k}=
 \begin{bmatrix}\widehat\Omega_{t,k}&0\\0&1\end{bmatrix}.
\]
Then update
\[
 \begin{gathered}
 A_{t,k+1}=\widetilde F_{t,k}\widetilde\Omega_{t,k}^{1/2},\qquad
 D_{t,k+1}=\widehat D_{t,k},\\
 f_{t,k}=\text{last column of }A_{t,k+1}.
 \end{gathered}
\]
The final block of $\widetilde\Omega_{t,k}^{1/2}$ is 1, so
$f_{t,k}=\widehat F_{2,t,k}$; $D_{t,k+1}$ remains separate.
We verified that, for every date $t$ and stage $k=0,\ldots,9$, the model
$A_{t,k+1}A_{t,k+1}^{\top}+D_{t,k+1}$ had strictly higher training log
likelihood than $A_{t,k}A_{t,k}^{\top}+D_{t,k}$.
After ten additions, $\widehat\Sigma_t=A_{t,10}A_{t,10}^\top+D_{t,10}$
has 17 common-factor columns. The figure labels stage $k$ as Factor $k+1$.

All stages on date $t$ use returns from January 2, 2020 through $t-1$,
with normalized exponential weights of half-life 126 trading days.
The sequence is seeded once on June 30, 2022 using Algorithm~\ref{alg:em}'s
initialization; subsequent dates use the warm start above.

\paragraph{The fitted loadings.}
The figure shows each new loading when it is added. For sector $S$, its share is
\[
 q_{t,k}(S)=\frac{\sum_{i\in S}f_{t,k,i}^2}{\sum_i f_{t,k,i}^2}.
\]
These are shares within each added factor. Subsequent stages can change
the covariance among earlier factors, so shares from different panels
should not be added together. \Cref{fig:additional-factor-shares} plots utilities,
real estate, information technology, and energy, together with the
disclosed AI basket.

We use the Global Industry Classification Standard (GICS) labels in the
\texttt{GICS Sector} field of the November 5, 2026 constituent snapshot.
Each security's label is held fixed across model dates. Sector shares
sum squared loadings over securities with that label, as in $q_{t,k}(S)$.
The separately specified AI basket overlaps the sector groups.

\begin{figure}
\centering
\includegraphics[width=\linewidth,height=0.75\textheight,keepaspectratio]{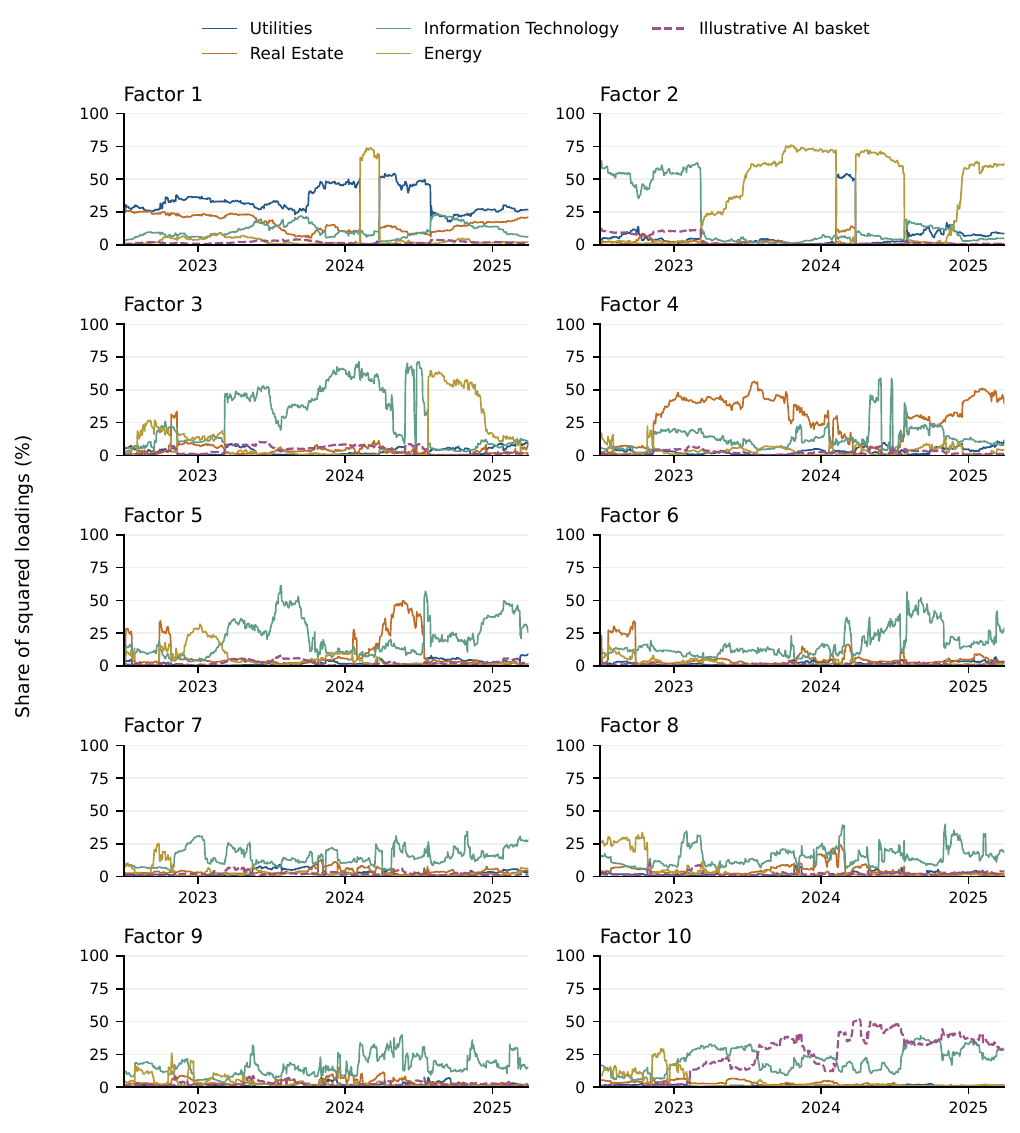}
\caption{Sector shares for ten sequentially added factors, with one
previous-day warm start per date and stage. All 689 model dates are
shown on a common 0--100\% scale. Only four of the eleven GICS sectors are
shown. Each panel uses the exposure estimated
at that addition stage, without smoothing or relabeling. Sector labels are
the fixed November 2026 GICS classifications described in the text. The AI basket
contains AMD, Amazon, Arista Networks, Broadcom, Alphabet (both share
classes), Meta, Microsoft, Micron, Nvidia, and Oracle. This is a disclosed
comparison group, not an exhaustive classification of AI beneficiaries;
it overlaps the sector categories. The basket and sector labels do not
enter estimation.}
\label{fig:additional-factor-shares}

\end{figure}
\subsection{Comparing Models with 73 Factors}

We return to the Barra short-term U.S. model used in the paper.
For each seed $b=3,\ldots,12$, we randomly remove two style factors and
five industry factors from the full 73-factor model. The removed set is
fixed across dates within each seed. We then apply Algorithm~\ref{alg:em} with
$n_1=66$ and $n_2=7$, giving the extended model the same total of 73
factors as full Barra.

On each date $t$, we run \cref{alg:em} using observations through $t$ and evaluate
the next day's return using equation~\eqref{eq:ltsq_r^2} and the subsequent prediction and $R^2$ definitions in \Cref{sec:r2_ret}. For each of 30
replications, both models use the same 90/10 asset split: the next-day
returns of the 90\% group predict those of the held-out 10\% group.
There are $T=1{,}135$ dates and $30T=34{,}050$ paired scores per seed.
For each seed, define
\[
 \begin{gathered}
 \Delta_{t,j}^{(b)}=R^2_{t,j,\mathrm{extended}}-R^2_{t,j,\mathrm{full\ Barra}},\\
 \bar\Delta_b=\frac{1}{30T}\sum_{t,j}\Delta_{t,j}^{(b)},
 \qquad
 \mathrm{SEM}_b=\frac{\operatorname{std}_{t,j}(\Delta_{t,j}^{(b)})}{\sqrt{30T}}.
 \end{gathered}
\]

\begin{table}[h!]
\centering
\footnotesize
\setlength{\tabcolsep}{3pt}
\renewcommand{\arraystretch}{1.12}
\begin{tabular}{@{}rp{0.39\linewidth}p{0.18\linewidth}rr@{}}
\toprule
Seed & Removed industries & Removed styles
     & $\bar\Delta_b$ & $\mathrm{SEM}_b$ \\
\midrule
3 & \texttt{AUTO, CNSTENG, COMMEQP, PRECMTLS, WIRELESS}
  & \texttt{BETANL, DIVYILD} & 0.0073 & 0.0003 \\
4 & \texttt{INSURNCE, OILGSDRL, SEMIEQP, SPLTYRET, TELECOM}
  & \texttt{DIVYILD, LEVERAGE} & 0.0062 & 0.0003 \\
5 & \texttt{AIRLINES, HOUSEDUR, MULTUTIL, PRECMTLS, REALEST}
  & \texttt{BETA, DIVYILD} & 0.0060 & 0.0003 \\
6 & \texttt{ELECEQP, GASUTIL, INSURNCE, TELECOM, TRADECO}
  & \texttt{BTOP, SIZE} & 0.0071 & 0.0003 \\
7 & \texttt{LIFEINS, MEDIA, OILGSCON, SOFTWARE, TRANSPRT}
  & \texttt{DIVYILD, SIZENL} & 0.0060 & 0.0003 \\
8 & \texttt{CNSTENG, COMMEQP, DISTRIB, OILGSDRL, TRANSPRT}
  & \texttt{BETA, RESVOL} & 0.0075 & 0.0003 \\
9 & \texttt{BEVTOB, CONGLOM, FOODRET, ROADRAIL, SPTYSTOR}
  & \texttt{LIQUIDTY, SIZE} & 0.0072 & 0.0003 \\
10 & \texttt{CNSTMATL, COMSVCS, PAPER, REALEST, SPTYCHEM}
   & \texttt{BETANL, GROWTH} & 0.0080 & 0.0003 \\
11 & \texttt{BIOLIFE, INDMACH, MEDIA, PSNLPROD, TELECOM}
   & \texttt{BETANL, EARNYILD} & 0.0060 & 0.0003 \\
12 & \texttt{APPAREL, COMPELEC, LIFEINS, SPTYSTOR, TELECOM}
   & \texttt{DIVYILD, GROWTH} & 0.0069 & 0.0003 \\
\bottomrule
\end{tabular}
\caption{Extended minus full Barra when both models have the same number of factors: mean paired difference and its
standard error. Factor labels follow the supplied Barra results.
Both numerical columns are in $R^2$ units.}
\label{tab:additional-barra}
\end{table}

Full Barra uses prescribed style and industry exposures. Our model keeps
the 66 retained exposures, learns seven vectors, and re-estimates
$\Omega$ and $D$. Thus the factor count is the same,
while the extended model has more freedom to adapt to recent returns.
The seven learned factors need not correspond to the seven named factors
removed.

\begin{sloppypar}
The average paired difference is positive for every seed, ranging from
0.005970 to 0.007951 (0.597--0.795 percentage points in $R^2$). It is shown in
\Cref{tab:additional-barra}, together with the standard error of the mean. Notice the paired mean difference 
is positive for all seeds and the 
1std range does not cross zero for any seed (although it has not been correctd for dependence across dates
and splits). The results suggest that the extended model has better out-of-sample $R^2$ than full Barra, even when both have the same number of factors.
These results compare the complete enhancement procedure, including
re-estimation of $\Omega$ and $D$, with full Barra at the same factor count.
\end{sloppypar}

% JOTA's alphabetical numbered style; bibliography entries remain upright.
\begingroup
\renewcommand{\emph}[1]{#1}
\providecommand{\doi}[1]{\url{https://doi.org/#1}}
\urlstyle{same}
\bibliographystyle{spmpsci}
\bibliography{references}

@article{markowitz1952,
  author = {Markowitz, Harry},
  title = {Portfolio {S}election},
  journal = {The Journal of Finance},
  year = {1952},
  volume = {7},
  number = {1},
  pages = {77--91},
  doi = {10.1111/j.1540-6261.1952.tb01525.x},
}

@book{grinold2000active,
  title = {Active {P}ortfolio {M}anagement},
  author = {Grinold, Richard C and Kahn, Ronald N},
  year = {2000},
  publisher = {McGraw-Hill},
  edition = {2nd},
  address = {New York},
}

@book{mcneil2015quantitative,
  title = {Quantitative {R}isk {M}anagement: {C}oncepts, {T}echniques and {T}ools---{R}evised {E}dition},
  author = {McNeil, Alexander J and Frey, R{\"u}diger and Embrechts, Paul},
  year = {2015},
  publisher = {Princeton University Press},
  address = {Princeton},
}

@article{sharpe1964capital,
  title = {Capital {A}sset {P}rices: A {T}heory of {M}arket {E}quilibrium under {C}onditions of {R}isk},
  author = {Sharpe, William F},
  journal = {The Journal of Finance},
  volume = {19},
  number = {3},
  pages = {425--442},
  year = {1964},
  doi = {10.1111/j.1540-6261.1964.tb02865.x},
}

@article{boyd2024markowitz,
  title = {Markowitz {P}ortfolio {C}onstruction at {S}eventy},
  author = {Boyd, Stephen and Johansson, Kasper and Kahn, Ronald and Schiele, Philipp and Schmelzer, Thomas},
  journal = {The Journal of Portfolio Management},
  volume = {50},
  number = {8},
  pages = {117--160},
  year = {2024},
  doi = {10.3905/jpm.2024.50.8.117},
}

@misc{msci_barra2023,
  author = {{MSCI Inc.}},
  title = {Equity {F}actor {M}odels},
  url = {https://www.msci.com/data-and-analytics/factor-investing/equity-factor-models},
  note = {Accessed 24 November 2026},
}

@article{engle1982autoregressive,
  title = {Autoregressive {C}onditional {H}eteroscedasticity with {E}stimates of the {V}ariance of {U}nited {K}ingdom {I}nflation},
  author = {Engle, Robert F},
  journal = {Econometrica},
  pages = {987--1007},
  year = {1982},
  volume = {50},
  number = {4},
  doi = {10.2307/1912773},
}

@article{johansson2023simple,
  title = {A {S}imple {M}ethod for {P}redicting {C}ovariance {M}atrices of {F}inancial {R}eturns},
  author = {Johansson, Kasper and Ogut, Mehmet G and Pelger, Markus and Schmelzer, Thomas and Boyd, Stephen},
  journal = {Foundations and Trends in Econometrics},
  volume = {12},
  number = {4},
  pages = {324--407},
  year = {2023},
  doi = {10.1561/0800000047},
}

@article{fan2016overview,
  title = {An {O}verview of the {E}stimation of {L}arge {C}ovariance and {P}recision {M}atrices},
  author = {Fan, Jianqing and Liao, Yuan and Liu, Han},
  journal = {The Econometrics Journal},
  volume = {19},
  number = {1},
  pages = {C1--C32},
  year = {2016},
  doi = {10.1111/ectj.12061},
}

@article{stuaricua2005nonstationarities,
  title = {Nonstationarities in {S}tock {R}eturns},
  author = {St{\u{a}}ric{\u{a}}, C{\u{a}}t{\u{a}}lin and Granger, Clive},
  journal = {Review of Economics and Statistics},
  volume = {87},
  number = {3},
  pages = {503--522},
  year = {2005},
  doi = {10.1162/0034653054638274},
}

@article{fama1992cross,
  title = {The {C}ross-{S}ection of {E}xpected {S}tock {R}eturns},
  author = {Fama, Eugene F and French, Kenneth R},
  journal = {The Journal of Finance},
  volume = {47},
  number = {2},
  pages = {427--465},
  year = {1992},
  doi = {10.1111/j.1540-6261.1992.tb04398.x},
}

@article{fama1993common,
  title = {Common {R}isk {F}actors in the {R}eturns on {S}tocks and {B}onds},
  author = {Fama, Eugene F and French, Kenneth R},
  journal = {Journal of Financial Economics},
  volume = {33},
  number = {1},
  pages = {3--56},
  year = {1993},
  doi = {10.1016/0304-405X(93)90023-5},
}

@article{famaFrench2015,
  author = {Fama, Eugene F and French, Kenneth R},
  title = {A {F}ive-{F}actor {A}sset {P}ricing {M}odel},
  journal = {Journal of Financial Economics},
  year = {2015},
  volume = {116},
  number = {1},
  pages = {1--22},
  doi = {10.1016/j.jfineco.2014.10.010},
}

@misc{frenchLibrary,
  author = {French, Kenneth R},
  title = {Data {L}ibrary},
  howpublished = {\href{https://mba.tuck.dartmouth.edu/pages/faculty/ken.french/data_library.html}{daily factor data},
    \href{https://mba.tuck.dartmouth.edu/pages/faculty/ken.french/Data_Library/f-f_5_factors_2x3.html}{five-factor definitions},
    \href{https://mba.tuck.dartmouth.edu/pages/faculty/ken.french/Data_Library/det_mom_factor_daily.html}{daily momentum}, and
    \href{https://mba.tuck.dartmouth.edu/pages/faculty/ken.french/Data_Library/det_st_rev_factor_daily.html}{daily short-term reversal}},
  note = {Accessed November 24, 2026},
}

@book{fabozzi2008handbook,
  title = {Handbook of {F}inance: {I}nvestment {M}anagement and {F}inancial {M}anagement},
  volume = {2},
  year = {2008},
  publisher = {Wiley},
  address = {Hoboken},
  editor = {Fabozzi, Frank J.},
  doi = {10.1002/9780470404324},
}

@book{hastie2009elements,
  title = {The {E}lements of {S}tatistical {L}earning: {D}ata {M}ining, {I}nference, and {P}rediction},
  author = {Hastie, Trevor and Tibshirani, Robert and Friedman, Jerome},
  year = {2009},
  edition = {2nd},
  publisher = {Springer},
  address = {New York},
  doi = {10.1007/978-0-387-84858-7},
}

@article{pafka2003noisy,
  title = {Noisy {C}ovariance {M}atrices and {P}ortfolio {O}ptimization {II}},
  author = {Pafka, Szil{\'a}rd and Kondor, Imre},
  journal = {Physica A: Statistical Mechanics and Its Applications},
  volume = {319},
  pages = {487--494},
  year = {2003},
  doi = {10.1016/S0378-4371(02)01499-1},
}

@article{ledoit2003honey,
  title = {Honey, {I} {S}hrunk the {S}ample {C}ovariance {M}atrix},
  author = {Ledoit, Olivier and Wolf, Michael},
  journal = {The Journal of Portfolio Management},
  volume = {30},
  number = {4},
  pages = {110--119},
  year = {2004},
  doi = {10.3905/jpm.2004.110},
}

@techreport{menchero2011barra,
  title = {The {B}arra {US} {E}quity {M}odel ({USE4}): {M}ethodology {N}otes},
  author = {Menchero, Jose and Orr, D. J. and Wang, Jun},
  year = {2011},
  type = {Model Insight},
  institution = {MSCI},
  month = {August},
}

@techreport{jp1996riskmetrics,
  title = {{RiskMetrics}---{T}echnical {D}ocument},
  author = {{J.P. Morgan} and {Reuters}},
  year = {1996},
  institution = {J.P. Morgan and Reuters},
  address = {New York},
  type = {Technical document},
  note = {4th edn.},
  url = {https://www.msci.com/documents/10199/5915b101-4206-4ba0-aee2-3449d5c7e95a},
}

@article{bai2008large,
  title = {Large {D}imensional {F}actor {A}nalysis},
  author = {Bai, Jushan and Ng, Serena},
  journal = {Foundations and Trends in Econometrics},
  volume = {3},
  number = {2},
  pages = {89--163},
  year = {2008},
  doi = {10.1561/0800000002},
}

@article{rubin1982algorithms,
  title = {{EM} {A}lgorithms for {ML} {F}actor {A}nalysis},
  author = {Rubin, Donald B and Thayer, Dorothy T},
  journal = {Psychometrika},
  volume = {47},
  number = {1},
  pages = {69--76},
  year = {1982},
  doi = {10.1007/BF02293851},
}

@inproceedings{seghouane2008iterative,
  title = {An {I}terative {P}rojections {A}lgorithm for {ML} {F}actor {A}nalysis},
  author = {Seghouane, Abd-Krim},
  booktitle = {{2008 IEEE Workshop on Machine Learning for Signal Processing}},
  pages = {333--338},
  year = {2008},
  editor = {Principe, J. and Erdogmus, D. and Adali, T. and Larsen, J.},
  publisher = {IEEE},
  address = {Piscataway},
  doi = {10.1109/MLSP.2008.4685502},
}

@misc{spector2024mosaic,
  title = {The {M}osaic {P}ermutation {T}est: {A}n {E}xact and {N}onparametric {G}oodness-of-{F}it {T}est for {F}actor {M}odels},
  author = {Spector, Asher and Barber, Rina Foygel and Hastie, Trevor and Kahn, Ronald N and Cand{\`e}s, Emmanuel},
  year = {2024},
  howpublished = {arXiv preprint arXiv:2404.15017},
  doi = {10.48550/arXiv.2404.15017},
}

@misc{yeon2025beyond,
  title = {Beyond {L}ow {R}ank: {F}ast {L}ow-{R}ank + {D}iagonal {D}ecomposition with a {S}pectral {A}pproach},
  author = {Yeon, Kingsley and Anitescu, Mihai},
  year = {2025},
  howpublished = {arXiv preprint arXiv:2512.17120},
  doi = {10.48550/arXiv.2512.17120},
}

@article{candes2025thematic,
  title = {Thematic {I}nvesting: {A} {R}isk-{B}ased {P}erspective},
  author = {Cand{\`e}s, Emmanuel and Hastie, Trevor and Hogan, Ked and Kahn, Ronald N and Luo, Robert and Spector, Asher},
  journal = {Financial Analysts Journal},
  volume = {81},
  number = {4},
  pages = {103--120},
  year = {2025},
  doi = {10.1080/0015198X.2025.2526483},
}

@book{mardia2024multivariate,
  author = {Mardia, K. V. and Kent, J. T. and Bibby, J. M.},
  title = {Multivariate {A}nalysis},
  year = {1979},
  publisher = {Academic Press},
  address = {London},
}

@article{menchero2008custom,
  title = {Custom {F}actor {A}ttribution},
  author = {Menchero, Jose and Poduri, Vijay},
  journal = {Financial Analysts Journal},
  volume = {64},
  number = {2},
  pages = {81--92},
  year = {2008},
  doi = {10.2469/faj.v64.n2.13},
}

@article{rosenberg1974extra,
  title = {Extra-{M}arket {C}omponents of {C}ovariance in {S}ecurity {R}eturns},
  author = {Rosenberg, Barr},
  journal = {Journal of Financial and Quantitative Analysis},
  volume = {9},
  number = {2},
  pages = {263--274},
  year = {1974},
  doi = {10.2307/2330104},
}

@article{lee2025narrative,
  title = {Narrative {F}actors and {R}isk {M}odels},
  author = {Lee, Wai and Brown, Ryan and de Silva, Harindra},
  journal = {The Journal of Portfolio Management},
  year = {2026},
  volume = {52},
  number = {3},
  pages = {243--254},
  doi = {10.3905/jpm.2025.1.786},
}

@article{fan2008high,
  title = {High {D}imensional {C}ovariance {M}atrix {E}stimation using a {F}actor {M}odel},
  author = {Fan, Jianqing and Fan, Yingying and Lv, Jinchi},
  journal = {Journal of Econometrics},
  volume = {147},
  number = {1},
  pages = {186--197},
  year = {2008},
  doi = {10.1016/j.jeconom.2008.09.017},
}

@techreport{nielsen2010fundamentals,
  title = {The {F}undamentals of {F}undamental {F}actor {M}odels},
  author = {Nielsen, Frank and Bender, Jennifer},
  number = {2010-24},
  year = {2010},
  institution = {MSCI Barra},
  type = {Research Paper},
  doi = {10.2139/ssrn.1707661},
}

@article{menchero2024evaluating,
  author  = {Menchero, Jose and Lazanas, Anthony},
  title   = {Evaluating and {C}omparing {R}isk {M}odel {P}erformance},
  journal = {The Journal of Portfolio Management},
  year    = {2024},
  volume  = {50},
  number  = {3},
  pages   = {90--110},
  doi     = {10.3905/jpm.2023.1.565}
}
\endgroup

\end{document}